\documentclass[fontsize=12pt,a4paper,headings=normal,twoside=false,leqno,parskip=half-,abstract=true]{scrartcl}
\pdfoutput=1
\usepackage[english]{babel}
\usepackage[utf8]{inputenc}
\usepackage{hyperref}
\hypersetup{
 pdftitle={Chaos in subcritical HL cosmologies},
 pdfauthor={Phillipo Lappicy},
 colorlinks=true,
 linkcolor=blue,
 citecolor=blue,
 filecolor=blue,
 urlcolor=blue}

\usepackage{graphicx} 
\usepackage[format=plain,labelfont=bf,font=small]{caption}
\usepackage{subfigure}
\usepackage{xcolor}
\usepackage[arrow, matrix, curve]{xy}
\usepackage{float}
\usepackage{caption}
\usepackage{tikz}
\usepackage{tikz-cd}
\usetikzlibrary{arrows, arrows.meta, calc, intersections, decorations.markings}
\tikzset{
decoration={markings,mark=at position 0.67 with {\arrow[thick,color=gray]{latex}}}
}

\usepackage{amsmath,amsthm}
\usepackage{amssymb} 
\usepackage{latexsym}
\usepackage{mathrsfs}

\usepackage{wasysym}

\DeclareOldFontCommand{\bf}{\normalfont\bfseries}{\mathbf}
\DeclareOldFontCommand{\it}{\normalfont\itshape}{\textit}

\usepackage{comment}

\newtheorem{thm}{Theorem}[section]
\newtheorem{lem}[thm]{Lemma}
\newtheorem{prop}[thm]{Proposition}

\usepackage[notref,notcite,color,final 
]{showkeys}

\newcommand{\RR}{\mathbb R}

\newcommand{\NN}{\mathbb N}

\newcommand{\A}{\mathrm{{A}}}

\newcommand{\KC}{\mathrm{K}^{\ocircle}}

\newcommand{\KM}{\mathcal{K}}

\renewcommand{\S}{\cal S}

\newcommand{\Q}{{\bf \mathrm{Q}}}
\newcommand{\T}{{ \mathrm{T}}}

\newcommand{\tab}{{\bf \mathrm{t}_{12}}}
\newcommand{\tac}{{\bf \mathrm{t}_{13}}}
\newcommand{\tba}{{\bf \mathrm{t}_{21}}}
\newcommand{\tbc}{{\bf \mathrm{t}_{23}}}
\newcommand{\tca}{{\bf \mathrm{t}_{31}}}
\newcommand{\tcb}{{\bf \mathrm{t}_{32}}}

\begin{document}

\title{{\LARGE{Periodic orbits in Ho\v{r}ava-Lifshitz cosmologies}}}

\author{
 \\
{~}\\
Kevin E. M. Church*, Olivier Hénot**, \\ 
Phillipo Lappicy***, Jean-Philippe Lessard**\\
and Hauke Sprink****
\vspace{2cm}}

\date{ }
\maketitle
\thispagestyle{empty}

\vfill

$\ast$\\
Universit\'e de Montr\'eal, Canada\\
$\ast\ast$\\
McGill University, Canada\\
$\ast\ast\ast$\\
Universidade Federal do Rio de Janeiro, Brazil\\
$\ast\ast\ast\ast$\\
Freie Universität Berlin,  Germany\\


\newpage
\pagestyle{plain}
\pagenumbering{arabic}
\setcounter{page}{1}

\begin{abstract}
\noindent We consider spatially homogeneous Ho\v{r}ava-Lifshitz (HL) models that perturb General Relativity (GR) by a parameter $v\in (0,1)$ such that GR occurs at $v=1/2$. 
We describe the dynamics for the extremal case $v=0$, which possess the usual Bianchi hierarchy: type $\mathrm{I}$ (Kasner circle of equilibria), type $\mathrm{II}$ (heteroclinics that induce the Kasner map) and type $\mathrm{VI_0},\mathrm{VII_0}$ (further heteroclinics). For type $\mathrm{VIII}$ and $\mathrm{IX}$, we use a computer-assisted approach to prove the existence of periodic orbits which are far from the Mixmaster attractor and thereby we obtain a new behaviour which is not described by the BKL picture of bouncing Kasner-like states.



\end{abstract}

\section{Introduction}\label{sec:intro}

\numberwithin{equation}{section}
\numberwithin{figure}{section}
\numberwithin{table}{section}

Ho{\v{r}}ava proposed a renormalizable, higher order derivative gravity theory that recovers general relativity (GR) in low energy but with improved high-energy behaviors, see \cite{hor09a,hor09b, HL_status_report}. This approach violates full spacetime diffeomorphism and introduces anisotropic scalings of space and time.
The deformation of the kinetics was firstly considered by DeWitt in \cite{DeWitt67}, whereas higher order derivatives in the potential was originally suggested by Lifshitz in \cite{Lifshitz41}.

More precisely, Ho{\v{r}}ava gravity is a gauge theory formulated in terms of a lapse $N$ and a shift vector $N^i$, which serve as Lagrange multipliers for the constraints in a Hamiltonian context, and a three-dimensional Riemannian metric $g_{ij}$ on the slices of the preferred foliation. We consider projectable theories, when the lapse depends only on time.
These objects arise from a 3+1 decomposition of a 4-metric according to,
\begin{equation}\label{genmetric}
\mathbf{g} = -N^2dt\otimes
dt + g_{ij}(dx^i + N^idt)\otimes (dx^j + N^jdt).
\end{equation}
In suitable units/scalings, the dynamics of Ho{\v{r}}ava vacuum gravity is governed by the
action
\begin{subequations}
\begin{equation}\label{action}
S = \int N\sqrt{ \det g_{ij}}({\cal T} - {\cal V}) dtd^3x,
\end{equation}
where ${\cal T}$ and ${\cal V}$ are given by
\begin{align}
{\cal T} &= K_{ij}K^{ij} - \lambda (K^k\!_k)^2,\label{kin}\\
{\cal V} &= {}^1 {\cal V}+{}^2 {\cal V}+{}^3 {\cal V}+{}^4 {\cal V}+{}^5 {\cal V}+{}^6 {\cal V}+\dots ,\\
&= k_1 R + k_2 R^2 + k_3 R^i\!_jR^j\!_i + k_4 R^i\!_jC^j\!_i + k_5 C^i\!_jC^j\!_i + k_6 R^3+ \dots\, .\label{calV}
\end{align}
\end{subequations}
Here $K_{ij}$ is the extrinsic curvature, $R$ and $R_{ij}$ are the scalar curvature
and Ricci tensor (of the spatial metric $g_{ij}$), respectively,
while $C_{ij}$ is the Cotton-York tensor~\cite{hor09b}, while $\lambda, k_1,\dots ,k_6$ are real parameters. 
Each potential term ${}^i {\cal V}$, where $i=1,\dots,6$, is defined as the $i^{th}$ summand in \eqref{calV}.
Repeated indices are summed over according to Einstein's
summation convention.

Full spacetime diffeomorphism invariance in GR fixes $\lambda=1$
uniquely and set all parameters of ${\cal V}$ in~\eqref{calV} to
zero, except $k_1=-1$ (i.e., ${\cal V} = -R$), see~\cite{hor09a,hor09b}.
Thus GR is a special case among the Ho{\v{r}}ava models.
The introduction of $\lambda$ changes the scaling properties of the field
equations, as does the introduction of additional curvature terms. Since some of
the curvature terms have different scaling properties, sums of such terms in ${\cal V}$
result in that the field equations no longer are scale-invariant.

The classical Belinski, Khalatnikov and Lifshitz (BKL) picture suggests that generic singularities in GR are: (i) \textit{vacuum dominated}, (ii) \textit{local} and (iii) \textit{oscillatory}.
In this regard, vacuum spatially homogeneous cosmologies, the Bianchi models, are expected to play a key role in the dynamical asymptotic behaviour, see \cite{bkl70,bkl82,Mixmaster,ugg13a,ugg13b}.
Similarly, the Bianchi models in Ho{\v{r}}ava gravity are also expected to describe generic singularities, see \cite{Bakas10,Kamenshchik,HellLappicyUggla}. 
In general, it is heuristically argued that there is an `asymptotically dominant' curvature term in \eqref{calV} toward the initial singularity, yielding a respective dominant Bianchi model, see \cite[Appendix A]{HellLappicyUggla}.

In Appendix \ref{app}, we deduce the dominant vacuum Bianchi models in Ho{\v{r}}ava gravity:
\begin{subequations}\label{full:subs}
	\begin{align}
	\Sigma_+^\prime  &= 4v(1-\Sigma^2)\Sigma_+ + \, {\cal S}_+, \label{full:Sigma+} \\
	\Sigma_-^\prime   &= 4v(1-\Sigma^2)\Sigma_- + \, {\cal S}_-, \label{full:Sigma-} \\
	N_1^\prime 			  &= -2(2v\Sigma^2 - 2\Sigma_+) \, N_1, \label{full:N1} \\
	N_2^\prime 			  &= -2(2v\Sigma^2 + \Sigma_+ + \sqrt{3}\Sigma_-)N_2, \label{full:N2} \\
	N_3^\prime 			  &= -2(2v\Sigma^2 + \Sigma_+ - \sqrt{3}\Sigma_- )N_3, \label{full:N3} 
	\end{align}
\end{subequations}
for some parameter $v \in [0,1]$, where the vector field is defined as follows
\begin{subequations}
	\begin{align}
	\Sigma^2  &	:= \Sigma_+^2 + \Sigma_-^2 ,\\
	\S_+ &  := 2 \left[ \left(N_3-N_2\right)^2 - N_1\left(2N_1-N_2-N_3\right) \right],\\
	\S_- &  := 2\sqrt{3} \left(N_3-N_2\right)\left(N_1-N_2-N_3\right).
	\end{align}
\end{subequations}
We denote ${}^\prime = d/d\tau$ the time derivative with respect to a time variable 
such that the singularity occurs as $\tau \rightarrow \infty$.
The evolution equations \eqref{full:subs} are bound to the following constraint which restricts the phase space $\RR^5$ to a four-dimensional invariant submanifold
\begin{equation} \label{constraint}
1 =  \Sigma^2 + \Omega_k.
\end{equation}
where $\Omega_k:=N_1^2 + N_2^2 + N_3^2 - 2N_1N_2 - 2N_2N_3 - 2N_3N_1$.

Equations \eqref{full:subs} describe vacuum spatially homogeneous models in GR when $v=1/2$. 
Moreover, in \cite[Appendix A.2]{HellLappicyUggla}, they argue that the equations \eqref{full:subs} are expected to asymptotically describe the dynamics of each individual potentials ${}^i {\cal V}$, where $i=1, \dots ,6$, for the respective parameters,
\begin{equation}\label{v's}
    {}^1v := \frac{1}{\sqrt{2(3\lambda-1)}}, \quad {}^{2}v = {}^{3}v = \frac{{}^1 v}{4}, \quad {}^4v = \frac{{}^1 v}{10}, \quad {}^{5}v = {}^{6}v = 0.
\end{equation}
For example, the exact equations for vacuum spatially homogeneous $\lambda$-R models arise for the parameter $v={}^1v$, whereas a Ho{\v{r}}ava model with only a cubic potential ${}^6 {\cal V}=k_6 R^3$ (i.e., with $k_1=\ldots=k_5=0,k_6\neq 0$) has dominant asymptotic equation with parameter $v={}^{6}v=0$. 

For $v=1/2$, the dynamics of equations \eqref{full:subs} has been extensively considered in the GR literature. A major achievement is the attractor theorem, which states that the $\omega$-limit set of generic solutions of Bianchi type $\mathrm{IX}$ is contained in the space of solutions of Bianchi type $\mathrm{I}$ and $\mathrm{II}$, also known as the Mixmaster attractor, see \cite{Ringstrom, heiugg09b}.
Therefore, it is expected that the concatenation of heteroclinic orbits of Bianchi type $\mathrm{II}$ and the induced map of Bianchi type $\mathrm{I}$ (the so-called \textit{Kasner map}) play a key role in the dynamics.
More rigorous results can be found, for example, in \cite{Beguin,Liebscher,Bernhard,Dutilleul}. A state~of~the~art overview is provided in \cite{Mixmaster}.

For $v\neq 1/2$, some features of GR persist, such as the Bianchi hierarchy of invariant sets. 
Type $\mathrm{I}$ consists of all $N_\alpha$ being zero, type $\mathrm{II}$ has a single non-zero $N_\alpha$, types $\mathrm{VI}_0,\mathrm{VII}_0$ have two non-zero $N_\alpha$, and Bianchi types $\mathrm{VIII},\mathrm{IX}$ consist of three non-zero $N_\alpha$, see \cite{HellLappicyUggla,LappicyDaniel}.

We now investigate the dynamics of the equations \eqref{full:subs} for the extremal case, $v=0$, which describes the asymptotics in case of a dominant cubic curvature term $k_6 R^3$ in \eqref{calV}:
\begin{subequations}\label{full:subsv=0}
	\begin{align}
	\Sigma_+^\prime  &= 2 \left[ \left(N_3-N_2\right)^2 - N_1\left(2N_1-N_2-N_3\right) \right], \label{full:Sigma+v=0}\\
	\Sigma_-^\prime   &= 2\sqrt{3} \left(N_3-N_2\right)\left(N_1-N_2-N_3\right), \label{full:Sigma-v=0}\\
	N_1^\prime 			  &= 4\Sigma_+ N_1, \label{full:N1v=0} \\
	N_2^\prime 			  &= -2( \Sigma_+ + \sqrt{3}\Sigma_-)N_2, \label{full:N2v=0} \\
	N_3^\prime 			  &= -2( \Sigma_+ - \sqrt{3}\Sigma_- )N_3, \label{full:N3v=0} 
	\end{align}
\end{subequations}
bound to the constraint \eqref{constraint}. Note that there is a conserved quantity given by 
\begin{equation}\label{eq:conserved_quantity}
    \Delta:= 3 \, |N_1 N_2 N_3|^{2/3} \qquad \text{ such that } \qquad \Delta'=0.
\end{equation}
The remaining of the paper describes the dynamics within the Bianchi hierarchy:
Section \ref{sec:typeIandII} constructs the types $\mathrm{I},\mathrm{II}$ with the induced Kasner map,
Section \ref{sec:typeVIandVII} reports on the types $\mathrm{VI_0},\mathrm{VII_0}$, and
Section \ref{sec:typeVIIIandIX} describes types $\mathrm{VIII},\mathrm{IX}$. 
Section \ref{sec:disc} possess concluding remarks.

\section{Bianchi type $\mathrm{I}$ and $\mathrm{II}$}\label{sec:typeIandII}

Bianchi type $\mathrm{I}$ solutions are characterized by all $N_\alpha = 0,\alpha=1,2,3$. The constraint \eqref{constraint} reduces the phase space to the  \textit{Kasner circle} of equilibria:
\begin{align} \label{defKC}
\KC := \left\{ \left(\Sigma_+,\Sigma_-, 0,0,0\right) \in \RR^5 \,\, |  \,\,
\Sigma_+^2 + \Sigma_-^2 = 1 \right\}.
\end{align}
There are three special points in $\KC$ corresponding to the Taub representation of Minkowski spacetime in GR. They are therefore called the \emph{Taub points} and given by
\begin{align}\label{Taubpoints}
\T_1 := \left(-1,0\right), \hspace{1.2cm}
\T_2 := \left(\dfrac{1}{2},\frac{\sqrt 3}{2}\right), \hspace{1.2cm} \T_3 := \left(\dfrac{1}{2},-\dfrac{\sqrt 3}{2}\right).
\end{align}
Note that the existence of the Kasner circle is independent on the parameter $v\in [0,1]$. Its stability, however, depends strongly on the parameter and this affects the type $\mathrm{II}$ solutions. 

In general, linearization of equation~\eqref{full:subsv=0} at $\mathrm{K}^{\ocircle}$ results in $N'_1 = 2\Sigma_+|_{\mathrm{K}^{\ocircle}}{N}_1$, and thereby the stability behaviour of $N_{1}$ changes when $\Sigma_+|_{\mathrm{K}^{\ocircle}} = 0$. We define the \emph{unstable Kasner arc}, denoted by $\mathrm{int}(A_{1})$, to be the points in $\mathrm{K}^{\ocircle}$ that are unstable in the $N_1$ variable, i.e., when $\Sigma_+ > 0$. The closure of $\mathrm{int}(A_{1})$ is denoted by
$A_1$ and is given by
\begin{equation}\label{A_1}
A_1:= \left\{ (\Sigma_+,\Sigma_-,0,0,0)\in \mathrm{K}^{\ocircle} \text{ $|$ }
\Sigma_+ \geq 0\right\}.
\end{equation}
Note that $A_1$ is a symmetric portion of
$\mathrm{K}^{\ocircle}$ with $\mathrm{Q}_1 := -\mathrm{T}_1$ in the middle. 
Equivariance yields the arcs $A_2,A_3$, where the respective variables $N_2,N_3$ are unstable. 
Define $\mathbf{A_{\alpha\beta}}:= A_\alpha\cap A_\beta$ and $\mathbf{A}:= {\cup_{\alpha\beta}} \, \mathbf{A_{\alpha\beta}}$,  for distinct $\alpha,\beta=1,2,3$.

Bianchi type $\mathrm{II}$ solutions consist of three disjoint hemispheres with a common boundary: the Kasner circle.
More precisely, it is the set of solutions where two $N$-variables are zero and one is nonzero, i.e. $\mathrm{II}_1 \cup \mathrm{II}_2 \cup \mathrm{II}_3$, where  
\begin{align}\label{BII_N_1}
\mathrm{II}_1 :=\left\{ \left( \Sigma_+,\Sigma_-, N_1,0,0\right) \in \RR^5 \,\, \Big|  \begin{array}{c}
N_1 > 0  \\
{N_1}^2 = 1-\Sigma^2
\end{array} \right\},
\end{align}
and $\mathrm{II}_2, \mathrm{II}_3$ are obtained by symmetry with a different non-zero $N$-variable.

Solutions of \eqref{full:subsv=0} in the hemisphere $\mathrm{II}_1$ are heteroclinics between two Kasner equilibria with $\alpha$-limit sets in $\text{int}(A_1)$ and $\omega$-limit in the complement $A_1^c:=\mathrm{K}^{\ocircle} \backslash A_1$.
Indeed, \eqref{full:subsv=0} becomes $\Sigma_-'=0$ (i.e., $\Sigma_-$ is constant) and $\Sigma_+'= -4(1-\Sigma^2)$, which implies that $\Sigma_+$ is monotonically decreasing for $\Sigma^2<1$. 
Similarly for $\mathrm{II}_2, \mathrm{II}_3$, see Figure \ref{KC_3d_map}. 
\vspace{-0.25cm}
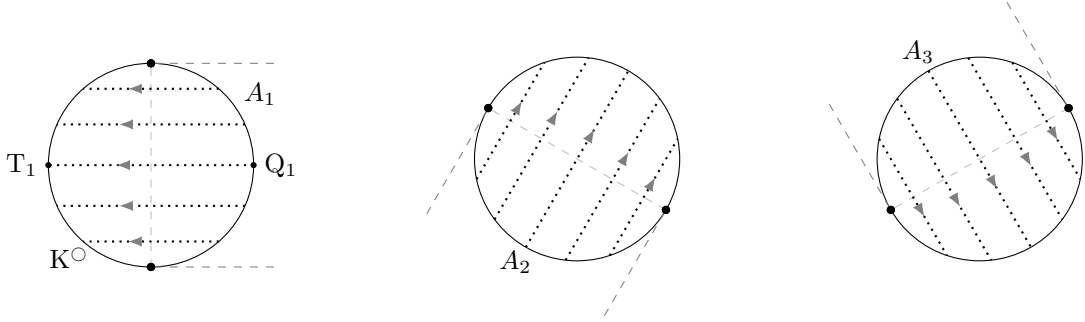
\begin{figure}[H]
	\centering
	\begin{tikzpicture}[scale=1.35]

\draw [lightgray,dashed] ( 0,-1) -- (0,1) ;

\draw (1,0) arc (0:360:1cm and 1cm);	 	
	
\draw  (0.8,0.7) circle (0pt) node [anchor=west]{\footnotesize $A_1$};

\filldraw [black] (-0.8,-0.7) circle (0pt) node[anchor=north] {\footnotesize $\KC$}; 	

\filldraw [black] (-1,0) circle (0.7pt) node[anchor=east] {\footnotesize $\T_1$};
\filldraw [black] (1,0) circle (0.7pt) node[anchor=west] {\footnotesize $\Q_1$};
	
\draw [gray, dashed, - ]  (1.2,1)--(0,1);
\draw [gray, dashed, - ]  (1.2,-1)--(0,-1);

\filldraw (0,-1) circle (1 pt);
\filldraw (0,1) circle (1 pt);
\draw (0,1) circle (1 pt);
\draw (0,-1) circle (1 pt);

\draw [dotted, thick, postaction={decorate}]  (1,0)--(-1,0);

\draw [dotted, thick, postaction={decorate}]  (0.9165,0.4)--(-0.9165,0.4);
\draw [dotted, thick, postaction={decorate}]  (0.9165,-0.4)--(-0.9165,-0.4);
	
\draw [dotted, thick, postaction={decorate}]  (0.66143,0.75)--(-0.66143,0.75);
\draw [dotted, thick, postaction={decorate}]  (0.66143,-0.75)--(-0.66143,-0.75);

\draw [white]  (0.2660,-1.5392) circle (0.01pt);

	\end{tikzpicture}
	\qquad\quad
	\begin{tikzpicture}[scale=1.35]

\draw [lightgray,dashed] ( 0.866,-0.5) -- (-0.866,0.5) ;

\draw (1,0) arc (0:360:1cm and 1cm);	 	

\draw  (-0.6, -0.8) circle (0pt) node [anchor=north]{\footnotesize $A_2$};

\draw [gray,dashed, - ]  (-1.4660,-0.5392)--(-0.8660,0.5);
\draw [gray,dashed, - ]  (0.2660,-1.5392)--(0.8660,-0.5);

\filldraw (0.8660,-0.5) circle (1 pt);
\filldraw (-0.8660,0.5) circle (1 pt);
\draw (0.8660,-0.5) circle (1 pt);
\draw (-0.8660,0.5) circle (1 pt);


\draw [dotted, thick, postaction={decorate}]  (-0.5000, -0.8660)--(0.5000,0.8660);

\draw [dotted, thick, postaction={decorate}]  (-0.8047, -0.5937)--(0.1118, 0.9937);
\draw [dotted, thick, postaction={decorate}]  (-0.1118, -0.9937)--(0.8047, 0.5937);

\draw [dotted, thick, postaction={decorate}]  (-0.9802, -0.1978)--(-0.3188, 0.9478);
\draw [dotted, thick, postaction={decorate}]  (0.3188, -0.9478)--(0.9802, 0.1978);

\filldraw (1.6,1.6) circle (0pt);
\filldraw (-1.6,1.6) circle (0 pt);
\filldraw (1.6,-1.6) circle (0 pt);
\filldraw (-1.6,-1.6) circle (0 pt);

	\end{tikzpicture}
	\qquad
        \begin{tikzpicture}[scale=1.35]
\draw [lightgray,dashed] ( 0.866,0.5) -- (-0.866,-0.5) ;

\draw (1,0) arc (0:360:1cm and 1cm);	 	

\draw  (-0.6, 0.84) circle (0pt) node [anchor=south]{\footnotesize $A_3$};

\draw [dashed, -,  gray]  (-1.4660, 0.5392)--(-0.866 , -0.5);
\draw [dashed, - , gray]  (0.2660,1.5392)--(0.866, 0.5);

\filldraw (0.866,0.5) circle (1 pt);
\filldraw (-0.866,-0.5) circle (1 pt);
\draw (0.866,0.5) circle (1 pt);
\draw (-0.866,-0.5) circle (1 pt);

\draw [dotted, thick, postaction={decorate}]  (-0.5000,0.8660)--(0.5000,-0.8660);

\draw [dotted, thick, postaction={decorate}]  (-0.8047, 0.5937)--(0.1118, -0.9937);
\draw [dotted, thick, postaction={decorate}]  (-0.1118, 0.9937)--(0.8047, -0.5937);

\draw [dotted, thick, postaction={decorate}]  (-0.9802, 0.1978)--(-0.3188, -0.9478);
\draw [dotted, thick, postaction={decorate}]  (0.3188, 0.9478)--(0.9802, -0.1978);

\filldraw (1.6,1.6) circle (0 pt);
\filldraw (-1.6,1.6) circle (0 pt);
\filldraw (1.6,-1.6) circle (0 pt);
\filldraw (-1.6,-1.6) circle (0 pt);

\draw [white]  (0.2660,-1.5392) circle (0.01pt);

\end{tikzpicture}
	\caption {Projection of the Bianchi type $\mathrm{II}$ heteroclinics in each hemisphere $\mathrm{II}_\alpha,\alpha=1,2,3,$ into the $\Sigma$-plane with $\alpha$-limits within int$(A_\alpha)$ and $\omega$-limits in $A_\alpha^c$.}\label{KC_3d_map}
\end{figure}
Thus, the heteroclinics in $\mathrm{II}_1$ induce a map from the $\alpha$-limit set to the $\omega$-limit set, denoted by $\KM_1:A_1 \rightarrow \overline{A_1^c}$. The map $\KM_1$ is a reflection along the $\Sigma_-$-axis given by the linear isometry
\begin{align}
\KM_1 \hspace{0.1cm} : \hspace{0.4cm} A_1 \hspace{0.5cm} &\longrightarrow \hspace{0.25cm}  \KC \setminus \overline{A_1} \\
\left(\Sigma_+, \Sigma_-\right) & \longmapsto 
\left( -\Sigma_+, \Sigma_-\right).
\end{align}
Analogous constructions in $\mathrm{II}_2,\mathrm{II}_3$ yield maps $\KM_2,\KM_3$.
Altogether,
they induce a map of the circle, 
called \emph{Kasner map} $\KM:\KC \rightarrow \KC $. 
Note that iterations of $\KM$ represent a heteroclinic chain formed by a sequence of Bianchi type $\mathrm{II}$ heteroclinic orbits.
%
Note that $\KM$ is multivalued in the set $\mathbf{A}$, 
whereas $\KM$ is uniquely determined in each arc $A_\alpha \setminus \mathbf{A}$, see Figure \ref{KC_0_map}.
The multivalued character for $v=0$ is reminiscent from the case $v\in (0,1/2)$, where $\KM$ can be formulated as a (non-hyperbolic discontinuous) skew-product dynamical system, see \cite{LappicyDaniel}.
\begin{figure}[H]
	\centering
	\begin{tikzpicture}[scale=1.35]

\draw [lightgray, dashed, - ]  (1.2,1)--(0,1);
\draw [lightgray, dashed, - ]  (1.2,-1)--(0,-1);

\draw [lightgray,dashed, - ]  (-1.4660,-0.5392)--(-0.8660,0.5);
\draw [lightgray,dashed, - ]  (0.2660,-1.5392)--(0.8660,-0.5);

\draw [dashed, -,  lightgray]  (-1.4660, 0.5392)--(-0.866 , -0.5);
\draw [dashed, - , lightgray]  (0.2660,1.5392)--(0.866, 0.5);

\draw [dashed, lightgray] ( 0,-1) -- (0,1) ;
\draw [dashed, lightgray] ( 0.866,-0.5) -- (-0.866,0.5) ;
\draw [dashed, lightgray] ( 0.866,0.5) -- (-0.866,-0.5) ;

\def \radius {1cm}
\draw (1,0) arc (0:360:1cm and 1cm);	 	

\node at ({60}:0.8cm) {\footnotesize $\mathbf{A}_{13}$};
\node at ({180}:0.75cm) {\footnotesize $\mathbf{A}_{23}$};
\node at ({-60}:0.75cm) {\footnotesize $\mathbf{A}_{12}$};

    \draw [ultra thick, domain=0:1.04,variable=\t,smooth] plot ({sin(\t r)},{cos(\t r)});
    \draw [rotate=120,ultra thick, domain=0:1.04,variable=\t,smooth] plot ({sin(\t r)},{cos(\t r)});
    \draw [rotate=240,ultra thick, domain=0:1.04,variable=\t,smooth] plot ({sin(\t r)},{cos(\t r)});

\draw  (0.95,0) circle (0pt) node [anchor=west]{\footnotesize $\A_1\setminus \mathbf{A}$};
\draw  (-0.6, -0.88) circle (0pt) node [anchor=north]{\footnotesize $\A_2\setminus \mathbf{A}$};
\draw  (-0.6, 0.85) circle (0pt) node [anchor=south]{\footnotesize $\A_3\setminus \mathbf{A}$};

	\end{tikzpicture}
	\hspace{-0.175cm}
	\begin{tikzpicture}[scale=1.35]
\draw [white]  (0.2660,-1.5392) circle (0.01pt);

    \draw [ultra thick, domain=0:1.04,variable=\t,smooth] plot ({sin(\t r)},{cos(\t r)});
    \draw [rotate=120,ultra thick, domain=0:1.04,variable=\t,smooth] plot ({sin(\t r)},{cos(\t r)});
    \draw [rotate=240,ultra thick, domain=0:1.04,variable=\t,smooth] plot ({sin(\t r)},{cos(\t r)});

\def \radius {1cm}
\draw (1,0) arc (0:360:1cm and 1cm);	 

\filldraw (0.55,0.9) circle (0.01 pt) node [anchor=west]  {$p$};
\draw (0.9928,0.11962) node[anchor=west] {};
\draw (-0.9928,0.11961524227) node[anchor=west] {};
\draw (-0.392829,-0.91961) node[anchor=west] {};

\draw[dotted, thick, postaction={decorate}] (0.6,0.8) -- (0.9928,0.1196);
\draw[dotted, thick, postaction={decorate}] (0.9928,0.1196) -- (-0.9928,0.1196);
\draw[dotted, thick, postaction={decorate}] (-0.9928,0.1196) -- (-0.392829,-0.91961);
\draw[dotted, thick, postaction={decorate}] (-0.392829,-0.91961) -- (0.6,0.8);

\draw[ dotted, thick, postaction={decorate}] (-0.9928,0.1196) -- (-0.6,0.8);
\draw[ dotted, thick, postaction={decorate}]  (0.6,0.8) -- (-0.6,0.8);
\draw[ dotted, thick, postaction={decorate}]  (-0.6,0.8) -- ( 0.39282,-0.91961);
\draw[ dotted, thick, postaction={decorate}]  ( 0.39282,-0.91961) -- (-0.392829,-0.91961);
\draw[ dotted, thick, postaction={decorate}]  ( 0.39282,-0.91961) -- (0.9928,0.11962);

\filldraw (1.2,1.2) circle (0 pt);
\filldraw (-1.2,1.2) circle (0 pt);
\filldraw (1.2,-1.2) circle (0 pt);
\filldraw (-1.2,-1.2) circle (0 pt);

	\end{tikzpicture}
	\quad
	\begin{tikzpicture}[scale=1.35]
\draw [white]  (0.2660,-1.5392) circle (0.01pt);

\def \radius {1cm}
\draw (1,0) arc (0:360:1cm and 1cm);	 	

    \draw [ultra thick, domain=0:1.04,variable=\t,smooth] plot ({sin(\t r)},{cos(\t r)});
    \draw [rotate=120,ultra thick, domain=0:1.04,variable=\t,smooth] plot ({sin(\t r)},{cos(\t r)});
    \draw [rotate=240,ultra thick, domain=0:1.04,variable=\t,smooth] plot ({sin(\t r)},{cos(\t r)});
    
\filldraw (0.866,0.5) circle (0.5 pt) node [anchor=west]  {};
\filldraw (-0.866,0.5) circle (0.5 pt)  node[anchor=west] {};
\filldraw (0,-1) circle (0.5 pt)  node[anchor=west] {};

\draw[dotted, thick, postaction={decorate}] (0.866,0.5) -- (-0.866,0.5) node [anchor = east] {\footnotesize $\tbc$};
\draw[dotted, thick, postaction={decorate}] (-0.866,0.5) -- (0,-1)node [anchor = north] {\footnotesize $\tab$};
\draw[dotted, thick, postaction={decorate}] (0,-1) -- (0.866,0.5)node [anchor = west] {\footnotesize $\tca$};

\filldraw (0.866,-0.5) circle (0.5 pt) node [anchor=west]  {};
\filldraw (-0.866,-0.5) circle (0.5 pt)  node[anchor=west] {};
\filldraw (0,1) circle (0.5 pt)  node[anchor=west] {};

\draw[dotted, thick, postaction={decorate}] (0.866,-0.5) -- (-0.866,-0.5) node [anchor = east] {\footnotesize $\tcb$};
\draw[dotted, thick, postaction={decorate}] (-0.866,-0.5) -- (0,1) node [anchor = south] {\footnotesize $\tac$};
\draw[dotted, thick, postaction={decorate}] (0,1) -- (0.866,-0.5) node [anchor = west] {\footnotesize $\tba$};

\end{tikzpicture}
    \,
	\begin{tikzpicture}[scale=1.35]
\draw [white]  (0.2660,-1.5392) circle (0.01pt);

    \draw [ultra thick, domain=0:1.04,variable=\t,smooth] plot ({sin(\t r)},{cos(\t r)});
    \draw [rotate=120,ultra thick, domain=0:1.04,variable=\t,smooth] plot ({sin(\t r)},{cos(\t r)});
    \draw [rotate=240,ultra thick, domain=0:1.04,variable=\t,smooth] plot ({sin(\t r)},{cos(\t r)});

\def \radius {1cm}
\draw (1,0) arc (0:360:1cm and 1cm);	 

\draw[dotted, thick, postaction={decorate}] (1,0) -- (-1,0);
\draw[dotted, thick, postaction={decorate}] (-1,0) -- (-0.5,0.866);
\draw[dotted, thick, postaction={decorate}] (-0.5,0.866) -- (0.5,-0.866);
\draw[dotted, thick, postaction={decorate}] (0.5,-0.866) -- (1,0);

\filldraw (1.2,1.2) circle (0 pt);
\filldraw (-1.2,1.2) circle (0 pt);
\filldraw (1.2,-1.2) circle (0 pt);
\filldraw (-1.2,-1.2) circle (0 pt);

\filldraw [black] (-1,0) circle (0.7pt) node[anchor=east] {\footnotesize $\T_1$};

	\end{tikzpicture}
	\caption {Left: All Kasner maps put together, where the map is multivalued in each (bold) region $\mathbf{A}$. 
	Middle-left: Example of all possible iterates of $p\in\mathbf{A}_{13}$, which returns to $p$ after a (multivalued) excursion.
	Middle-right: The tangential points (i.e., the boundary of the bold region) form two period three orbits. The existence of heteroclinic chains of period three was shown in \cite{Kamenshchik} using the Kasner parameter $u\in [1,\infty]$.
	Right: An example of a period four, which contains the Taub point $\mathrm{T}_1$. } \label{KC_0_map}
\end{figure}
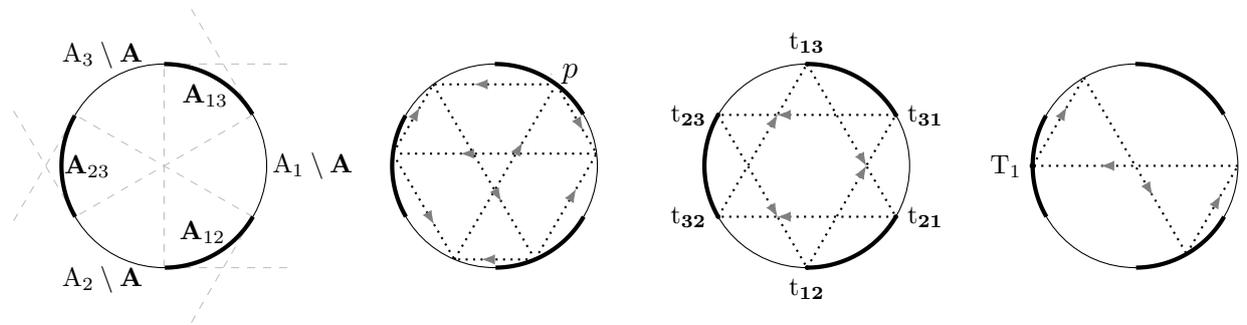
%
%

Note that there are different orbits starting at the point $p\in \mathbf{A}_{13}$ in Figure \ref{KC_0_map} (middle-left). First, there is an orbit of period four through the maps $\KM_2\circ \KM_3\circ \KM_1\circ \KM_3$. Second, there is a period six orbit through the maps $\KM_2 \circ \KM_1 \circ \KM_3 \circ \KM_2 \circ \KM_1 \circ \KM_3$. 
These orbits are not unique, i.e., there are other orbits of period four and six starting at $p$ through a different combination of maps.
Thus, tracking different combinations of period four and six orbits generate periodic orbits of any period $4j+6k$ for $j,k \in \NN$. 
Moreover, there are also several aperiodic chains by alternating among the orbits of period four and six while increasing the number of appearances of each chain (i.e., the first chain appears once, the second appears twice, the first chain appears three times, the second appears four times, etc).
%
Therefore, each point $p\in \KC$ generates a plethora of multivalued dynamical possibilities.
Next, we describe the orbit structure and discrete dynamics of the map $\KM$.
\begin{lem} \label{Lemma_zero}
Each orbit of the Kasner map $\KM$ contains at most six different points such that each of the six arcs of $\KC$ 
contains at most one point. 
Moreover, 
	\begin{itemize}
		\item[(i)] An orbit with exactly three points consists of a heteroclinic chain of minimal period three which is composed solely of tangential points.
		\item[(ii)]  An orbit with exactly four points consists of heteroclinic chains of minimal period four.
		\item[(iii)] There are no orbits with only five points.
		\item[(iv)] 
		Consider a point in $\KC$ that has six points in its orbit. Thus, starting at this point, there are periodic orbits of any even period $2n \ge 4,n\in\mathbb{N},$ and aperiodic orbits.
	\end{itemize}
\end{lem}
\begin{proof}

Note that the restricted map $\KM|_{A_\alpha \setminus \mathbf{A}}$ is a bijection in each of the six arcs of $\KC$, in contrast to the case $v>0$.
Moreover, it takes at least four iterates of $\KM$ to re-enter the same region of the six arcs again, except for the tangential points which consist of the boundary of the multivalued region $\mathbf{A}$. This proves $(i)$ and $(ii)$. 

Next, we prove that any point in the circle has at most six points in its orbit, where the points lie in different arcs.
After at most six iterations of $\KM$, one of the six arcs is revisited by the pigeonhole principle. Moreover, between the first and the second appearances of such a revisited arc, the map $\KM$ is iterated by an even number of times $n\leq 6$. 
Denote such iterate by $\KM^n:= \KM_{\omega_n} \circ ... \circ \KM_{\omega_1}$ for some $\omega_k\in \{1,2,3\}$ where $k=1,\ldots,n$. Note that $\KM^n$ is the identity map, since each iterate of $\KM$ is an isometry which is a one-to-one correspondence between two arcs.
Hence $\KM^n$ amounts to a rotation of each arc, which proves the claim. 

In case that an orbit has five points, note that they must be in different arcs. 
By symmetry, either the image or pre-image of these five points must also contain a sixth point in the remaining arc (that had no point from the original five arcs). 
This proves $(iii)$.

Lastly, consider an orbit with six points.
Due to symmetry, there are three points in the multivalued region $\mathbf{A}$ 
and three points in the single valued region. 
Thus, all orbits have the same structure as in Figure \ref{KC_0_map} (middle-left).
In particular, there are periodic orbits of any period $4j+6k$ for $j,k \in \NN$ by concatenating different period four and six orbits. 
Next, we show that this combination is able to generate periodic orbits of any period $2n\geq 4$, i.e., for any $n\in\mathbb{N}$, there are $k=k(n),j=j(n)$ such that $n=2j+3k$. There are two cases depending on the parity of $n$.
Indeed, if $n$ is even (resp. odd), then choosing $k$ even (resp. odd) implies that $n-3k$ is even in both cases. Thus, for any $k\geq 0$
with same parity as $n$ such that $n-3k\geq 0$, one can define $j:=(n-3k)/2$. 
Aperiodic chains occur by alternating orbits of period four and six while increasing the number of appearances of each orbit.
\end{proof}

%
%

Recall that the map $\KM$ has different qualitative regimes depending on $v$.
For $v=1/2$, the map $\KM$ is known to be generically chaotic, see~\cite{bkl70,khaetal85,ugg13a,ugg13b} and references therein. 
For $v\in (1/2,1)$, the map $\KM$ is chaotic in a Cantor set of measure zero, see \cite{HellLappicyUggla}. 
For $v\in(0,1/2)$, the (multivalued) map $\KM$ is chaotic, see \cite{LappicyDaniel}. Thus, $v=1/2$ corresponds to a bifurcation from non-generic to generic chaos. The case $v=0$ also corresponds to a bifurcation. However, the map $\KM$ cannot be chaotic for $v=0$, since it is given by a sequence of linear isometries. Indeed, there is no sensitivity to initial conditions, since any open set $U\subseteq \KC$ has six sets of same length as possible iterates. For the same reason, there is no topological mixing. The only property of chaos that holds for $v=0$ is the density of periodic orbits. 

\section{Bianchi type $\mathrm{VI}_0$ and $\mathrm{VII}_0$}\label{sec:typeVIandVII}
To obtain the equations for the type $\mathrm{VI}_0$ and $\mathrm{VII}_0$ models we set, without loss of generality, 
$N_1=0$, $N_2>0$, $N_3<0$ for type $\mathrm{VI}_0$,
and $N_1=0$, $N_2>0$, $N_3>0$ for type $\mathrm{VII}_0$,%
\begin{subequations}\label{dynsyslambdaRVIVII}
\begin{align}
\Sigma_+^\prime &= 2(1-\Sigma_+^2 - \Sigma_-^2),\label{SpVIVII}\\
\Sigma_-^\prime &= 2\sqrt{3}(N_2^2 - N_3^2),\label{SmVIVII}\\
N_2^\prime &= -2(\Sigma_+ + \sqrt{3}\Sigma_-)N_2,\label{N2primeVIVII}\\
N_3^\prime &= -2(\Sigma_+ - \sqrt{3}\Sigma_-)N_3,
\end{align}
\end{subequations}
and the constraint $1 - \Sigma_+^2 - \Sigma_-^2 - (N_2 - N_3)^2 =0$.
%
%
%
%

Due to the constraint, the state spaces for the
type $\mathrm{VI}_0$ and $\mathrm{VII}_0$ models with $N_1=0$ are
3-dimensional with a 2-dimensional boundary given by the union of the
invariant type $\mathrm{II}_2$, $\mathrm{II}_3$ and $\mathrm{K}^\ocircle$
sets. 
The type $\mathrm{VI}_0$ and $\mathrm{VII}_0$ models also share the region $A_2\cap A_3\subseteq \mathrm{K}^\ocircle$ in its boundary 
such that $N_2,N_3$ are unstable in $\mathrm{int}(A_2\cap A_3)$.
The stable set in $\mathrm{K}^\ocircle$ is
given by $S_{\mathrm{VI}_0,\mathrm{VII}_0} := \mathrm{K}^\ocircle\backslash
\mathrm{int}(A_2 \cup A_3)$.
See Figure~\ref{KC_VI_VII}.
Note that the type $\mathrm{VI}_0$ has a relatively compact state space,
whereas type $\mathrm{VII}_0$ has an unbounded one.
Indeed, the constraint implies that $\Sigma_+^2 + \Sigma_-^2\leq 1$.
For type $\mathrm{VI}_0$, note that $(N_2-N_3)^2= N_2^2 + N_3^2 + 2|N_2N_3|$, and
the constraint yields $N_2^2 \leq 1 -\Sigma^2$ and
$N_3^2 \leq 1 -\Sigma^2$, where the equalities hold for the $\mathrm{II}_2$
and $\mathrm{II}_3$ boundary sets, respectively. For type $\mathrm{VII}_0$,
introducing $N_\pm := N_2\pm N_3$ imply that the
constraint can be written as $\Sigma^2 + N_-^2 = 1$,
and thus $\Sigma_\pm$ and $N_-$ are bounded, while $N_+$ is unbounded.
\vspace{-0.25cm}
\begin{figure}[H]
	\centering
	\begin{tikzpicture}[scale=1.35]


\draw [lightgray,dashed, - ]  (-1.4660,-0.5392)--(-0.8660,0.5);
\draw [lightgray,dashed, - ]  (0.2660,-1.5392)--(0.8660,-0.5);

\draw [dashed, -,  lightgray]  (-1.4660, 0.5392)--(-0.866 , -0.5);
\draw [dashed, - , lightgray]  (0.2660,1.5392)--(0.866, 0.5);


\def \radius {1cm}
\draw (1,0) arc (0:360:1cm and 1cm);	 	

\node at ({180}:0.5cm) {\footnotesize $A_2\cap A_3$};

    \draw [rotate=120,ultra thick, domain=0:1.04,variable=\t,smooth] plot ({sin(\t r)},{cos(\t r)});
    
    \draw [rotate=300,ultra thick, dotted, white, domain=0:1.04,variable=\t,smooth] plot ({sin(\t r)},{cos(\t r)});

\draw  (0.95,0) circle (0pt) node [anchor=west]{\footnotesize $S_{\mathrm{VI}_0,\mathrm{VII}_0}$};
\draw  (-0.6, -0.875) circle (0pt) node [anchor=north]{\footnotesize $\A_2\setminus A_2\cap A_3$};
\draw  (-0.6, 0.875) circle (0pt) node [anchor=south]{\footnotesize $\A_3\setminus A_2\cap A_3$};

\draw  (0,-2.375) circle (0pt);

	\end{tikzpicture} 
	\hspace{-0.45cm}
	\includegraphics[scale=0.15, trim = 1cm 0cm 1cm 0cm, clip]{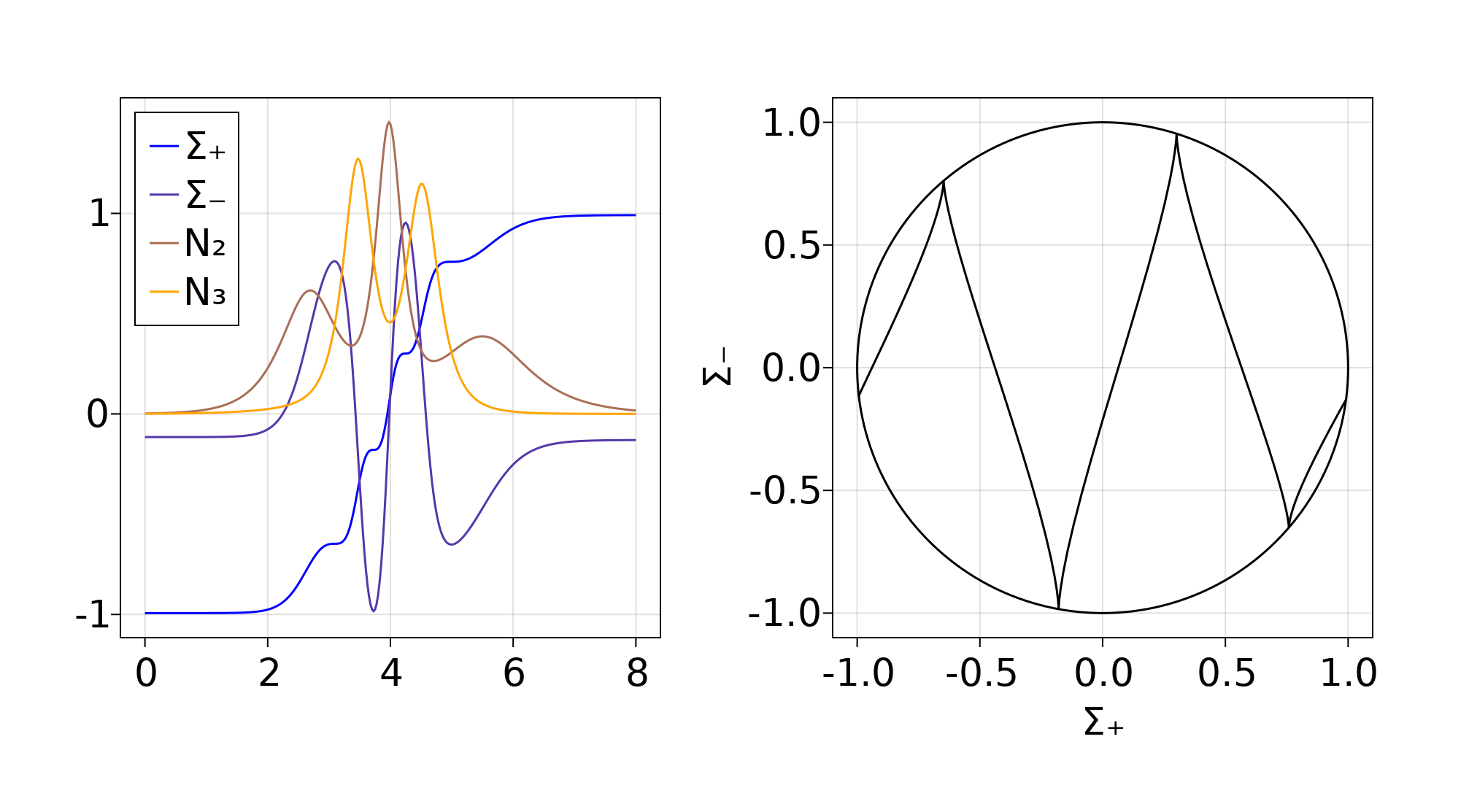}




    

\vspace{-0.25cm}
	\caption {Left: The stability of $\KC$ in type $\mathrm{VI}_0,\mathrm{VII}_0$. The (bold) set $A_2\cap A_3$ has two unstable eigenvalues, 
	$A_2\setminus A_2\cap A_3$ and $A_3\setminus A_2\cap A_3$ have one unstable eigenvalue, 
	and the (dashed) set $S_{\mathrm{VI}_0,\mathrm{VII}_0}$ is stable. 
	Middle: An example of a type $\mathrm{VII}_0$ heteroclinic with initial data $(\Sigma_+,\Sigma_-,N_2,N_3)=(-0.993,-0.115,0.002,0.001)$.
	Right: The projection of such example onto the $\Sigma$-plane. 
	} \label{KC_VI_VII}
\end{figure}
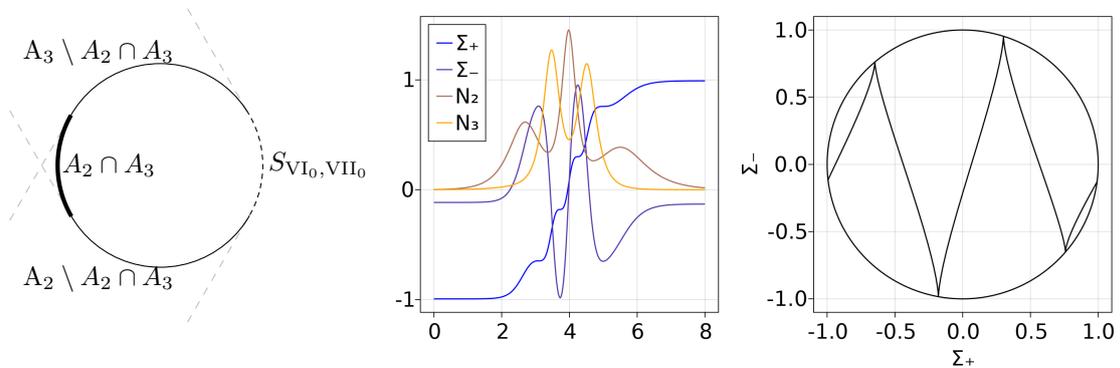
Next, we prove that the solutions of \eqref{dynsyslambdaRVIVII} are heteroclinics and we describe their $\alpha,\omega$-limit sets, following the lines of \cite{HellLappicyUggla}.
\begin{prop}\label{lem:genVIalpha}
In Bianchi type $\mathrm{VI}_0$ the $\alpha$-limit set for all orbits resides in $A_2\cap A_3\subseteq \mathrm{K}^\ocircle$.
The $\omega$-limit set for all orbits resides in the set $S_{\mathrm{VI}_0}$.
\end{prop}
\begin{proof}
All type $\mathrm{VI}_0$ orbits satisfy $\Sigma^2 < 1$, and thereby $|\Sigma_+|<1$, while
$\Sigma^2 = 1$ corresponds to $\mathrm{K}^\ocircle$,
since the constraint yields $(N_2-N_3)^2=0$, 
and thus $N_2=N_3=0$.

Equation \eqref{dynsyslambdaRVIVII} implies that $\Sigma_+$ is monotonically increasing, except when $\Sigma^2=1$. Thus $\lim_{\tau\rightarrow\pm\infty}\Sigma^2 = 1$, and
hence $\lim_{\tau\rightarrow\pm\infty}(N_2,N_3) = (0,0)$, due to the constraints.
Therefore both the $\alpha$- and $\omega$-limit sets for all type $\mathrm{VI}_0$
orbits belong to the set $\mathrm{K}^\ocircle$. It then follows from the stability
properties of $\mathrm{K}^\ocircle$ that the $\alpha$-limit set (resp. $\omega$-limit set) for these orbits
resides in the set $A_2\cap A_3$ (resp. $S_{\mathrm{VI_0}}$).
\end{proof}

Let us now turn to type $\mathrm{VII}_0$, but before presenting asymptotic results
we first consider the locally rotationally symmetric (LRS) type $\mathrm{VII}_0$ subset. This invariant set
is given by $N_-=0$ and $\Sigma_-=0$, where the constraint
divides the LRS subset into two disjoint invariant sets consisting of the two
lines at $\Sigma_+ = 1$ and $\Sigma_+ = -1$, i.e.,
\begin{subequations}\label{LRS}
\begin{align}
\mathrm{LRS}^\pm &:= 
\left\{ (\Sigma_+,0,N_2,N_3) \in \mathbb{R}^4 \Bigm|
\begin{array}{c}
\,\, \Sigma_+=\pm 1, \\
\,\, N_2=N_3\neq 0
\end{array}
\right\},
\end{align}
\end{subequations}
where the superscript of $\mathrm{LRS}^\pm$ is determined by the sign of $\Sigma_+$.
Let $N := N_2 = N_3>0$. Then the flow on the $\mathrm{LRS}^\pm$ subsets is determined by
\begin{equation}\label{Neqn}
N^\prime = -2\Sigma_+ N, \qquad \Sigma_+ = \pm 1.
\end{equation}
On $\mathrm{LRS}^+$, the variable $N\in (0,\infty)$
decreases from  $\lim_{\tau\rightarrow -\infty}N =\infty$  to $0$, and hence
the orbit in the invariant line ends at $\mathrm{Q}_1\in \mathrm{K}^\ocircle$. On $\mathrm{LRS}^-$, there is an orbit that emanates from
$\mathrm{T}_1$, where $N\in (0,\infty)$ subsequently increases,
which results in $\lim_{\tau\rightarrow\infty}N = \infty$.
%
%
%
%
\begin{prop}\label{lem:genVIVIIomega}
In Bianchi type $\mathrm{VII}_0$ 
the $\alpha$-limit set for all orbits reside in $A_2\cap A_3 \subseteq \mathrm{K}^\ocircle$, apart from the $\mathrm{LRS}^+$ set which is an
orbit such that $\lim_{\tau\rightarrow-\infty} N = \infty$, where $N:=N_2=N_3$.
The $\omega$-limit set for all orbits resides in the stable set $S_{\mathrm{VII}_0}\subseteq \mathrm{K}^\ocircle$, apart from the $\mathrm{LRS}^-$ set which is an orbit such that $\lim_{\tau\rightarrow\infty} N = \infty$.
\end{prop}
\begin{proof}
The exceptions follow from the previous analysis of the LRS type $\mathrm{VII}_0$ subset, due to~\eqref{Neqn}. Consider therefore type
$\mathrm{VII}_0$ non-LRS orbits, i.e., orbits for which $\Sigma_-^2 + N_-^2>0$
and thereby $|\Sigma_+|<1$ due to the constraint. Note that in contrast to the type $\mathrm{VII}_0$ unbounded state space, its boundary is given by the compact set $\mathrm{II}_2\cup\mathrm{II}_3\cup\mathrm{K}^\ocircle$.

First we prove the result for $\omega$-limit sets.
Note that 
\begin{equation}\label{monp}
\Sigma_+^\prime = 2N_-^2,\qquad
\Sigma_+''|_{N_-=0} = 0,\qquad
\Sigma_+'''|_{N_-=0} = 24(N_2+N_3)^2\Sigma_-^2.
\end{equation}
Thus $\Sigma_+$ is increasing for all non-LRS orbits (i.e., orbits such that $\Sigma_-^2 + N_-^2>0$), except when $N_-=0$ (and thereby
$\Sigma_-\neq 0$), which corresponds to an inflection point in the growth of the positive quantity $\Sigma_+$, due to~\eqref{monp}. Thus all non-$\mathrm{LRS}^-$ orbits eventually enter the (positively) invariant set
$\Sigma_+>0$. 

Moreover, the Lyapunov function $Z_{\mathrm{sub}} := 1/|N_2N_3|>0$ satisfy $Z_{\mathrm{sub}}^\prime = 4 \Sigma_+ Z_{\mathrm{sub}}$. 
Thus, $Z_{\mathrm{sub}}>0$ is monotonically increasing in the invariant set $\Sigma_+>0$, for all non-$\mathrm{LRS}$ orbits. It follows that $\lim_{\tau\rightarrow\infty}Z_{\mathrm{sub}} = \infty$ and thereby
$\lim_{\tau\rightarrow\infty}N_2N_3 = 0$. Thus the $\omega$-limit set of all
non-$\mathrm{LRS}^-$ orbits resides in
the $\mathrm{II}_2\cup\mathrm{II}_3\cup\mathrm{K}^\ocircle$ boundary set. The same local analysis of this boundary set as in type $\mathrm{VI}_0$ yields the result for the non-$\mathrm{LRS}^-$ orbits in type $\mathrm{VII}_0$.

Next we prove the result for $\alpha$-limit sets.
Similar arguments as in the previous discussion about $\omega$-limit sets lead to the following: $\Sigma_+$ is monotonically decreasing when $\tau\rightarrow - \infty$, which shows that the $\alpha$-limit set for all non-$\mathrm{LRS}^+$ orbits resides in the set $A_2\cap A_3$.
\end{proof}
%



\section{Bianchi type $\mathrm{VIII}$ and $\mathrm{IX}$}\label{sec:typeVIIIandIX}

We prove the existence of several periodic orbits, some of which are far from the Mixmaster attractor consisting of Bianchi type $\mathrm{I}$ and $\mathrm{II}$, see \cite{Ringstrom, heiugg09b}.  
This yields a behaviour which is not described by the BKL picture.
We structure our main result in the following theorem.
%
%
\begin{thm}\label{thm:typeVIII_IX}
The differential equation \eqref{full:subsv=0} possesses nontrivial periodic orbits along which the constraint \eqref{constraint} is satisfied. In particular, there are periodic solutions along which the product $N_1 N_2 N_3$ is constant and equal to $10^k$, for each $k\in\{1,0,-1,\dots,-7\}$.
\end{thm}
%

Theorem \ref{thm:typeVIII_IX} is proven by means of a computer-assisted approach. The implementation of the computer-assisted proofs is done in Julia (cf. \cite{Julia}) with the packages \textit{RadiiPolynomial.jl} (cf. \cite{RadiiPolynomial.jl}) and \textit{IntervalArithmetic.jl} (cf. \cite{IntervalArithmetic.jl}). The code for the computer-assisted proofs may be found at \cite{code_cite}. The details of the method are provided in Section \ref{sec:typeVIIandIX_Proofs}. 

Recall that \eqref{full:subsv=0} has two conserved quantities: the product $N_1N_2N_3$ in \eqref{eq:conserved_quantity}, and the right-hand side of the constraint $\Sigma^2 + \Omega_k$ in \eqref{constraint}. 
A consequence of the computer-assisted proof is that there exists a two-parameter family of periodic orbits defined in a neighbourhood of \textit{each} of the orbits in Theorem \ref{thm:typeVIII_IX}, parameterized by $(\Sigma^2 + \Omega_k,N_1N_2N_3)$ nearby $(1,10^k)$, for $k\in\{1,0,-1,\dots,-7\}$. 
We have not computed the size of these existence neighbourhoods.

Furthermore, we numerically computed the Floquet multipliers of each periodic orbit in Theorem \ref{thm:typeVIII_IX}; the numerical integrations of the monodromy matrices were done in Julia with the package \textit{DifferentialEquations.jl} (cf. \cite{DifferentialEquations.jl}).
We expect that there are three multipliers equal to unity, due to the two conserved quantities $(\Sigma^2 + \Omega_k,N_1N_2N_3)$ and the translation-invariance of the periodic orbit, which is in agreement with our numerical results.
For the two remaining multipliers, we consistently found one stable (absolute value less than one) and one unstable (absolute value greater than one) multiplier, which should influence the local dynamics. 
It is of interest to describe the forward dynamics of solutions within the unstable manifold of these periodic orbits, as we do not exclude the existence of homoclinic or heteroclinics between them. 
In particular, the global attractor for $v=0$ must be different than the conjectured Mixmaster attractor of type $\mathrm{I}$ and $\mathrm{II}$ solutions for $v>0$ in \cite{HellLappicyUggla}.

The projection of several orbits into the $\Sigma$-plane are plotted in Figure \ref{fig:TypeVIII_continuation} and Figure \ref{fig:TypeIX_continuation}. 
To further assist in visualizing the periodic orbits in the $\Sigma$-plane as the conserved quantity $N_1N_2N_3$ decreases from 10 to $10^{-7}$, we have computed a numerical continuation of the periodic orbits, together with their periods, through the entire range $N_1N_2N_3\in[10^{-7},10]$. 
We have used \textit{Makie.jl} \cite{Makie.jl} for visualization. While we believe that the continuum of periodic orbits visible in the figures could be proven using validated continuation methods (via the uniform contraction theorem, e.g. see \cite{val_cont1,val_cont2}), we have made no attempt to do this here. 
Time series plots for some orbits 
are available in Figure \ref{fig:TypeVIII_TimeSeries} and Figure \ref{fig:TypeIX_TimeSeries}.
\vspace{-0.05cm}
\begin{figure}[H]
\centering
\includegraphics[scale=0.15, trim = 1cm 0cm 1cm 0cm, clip]{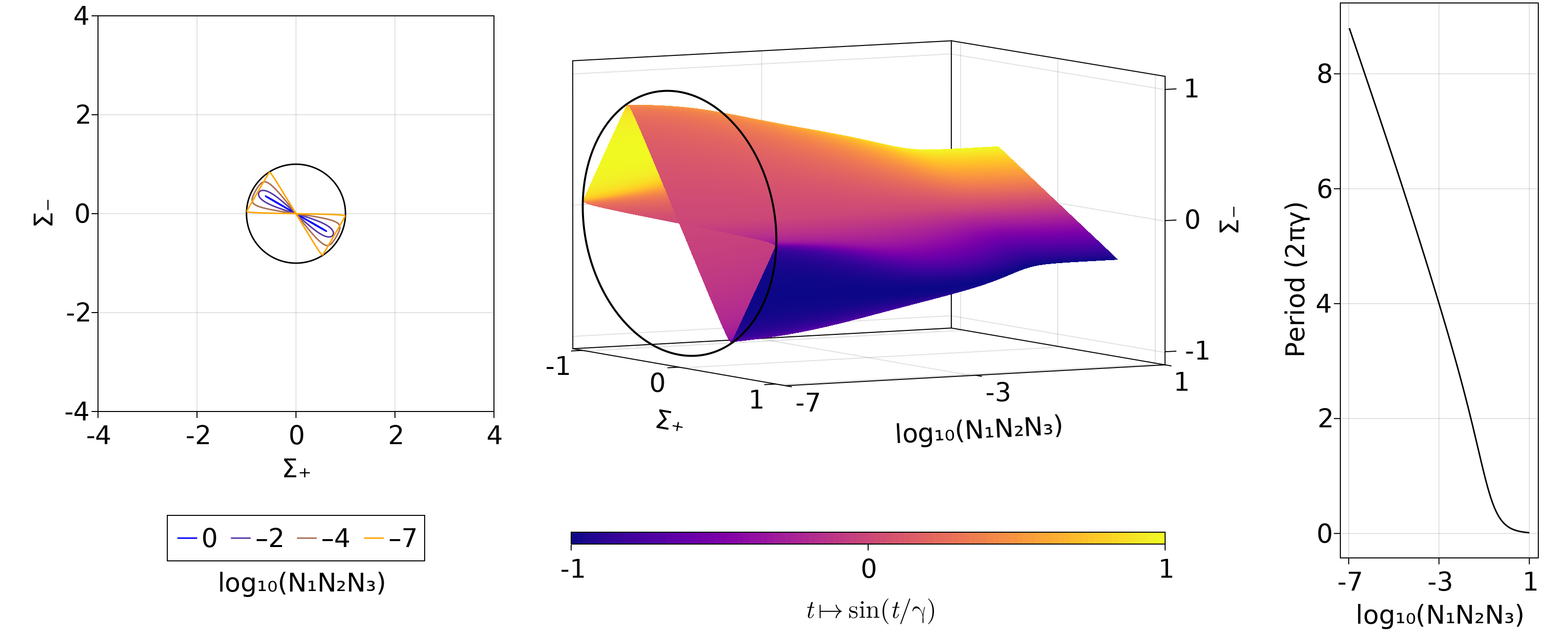}
\vspace{-0.6cm}
\caption{
Left: Projection into the $\Sigma$-plane of some Bianchi type $\mathrm{VIII}$ periodic orbits of Theorem \ref{thm:typeVIII_IX}. 
Middle: Numerical continuation of the Bianchi type $\mathrm{VIII}$ periodic orbits, 
for $N_1N_2N_3=10^k$, $k\in[-7,1]$. 
For each $k=\log_{10}(N_1N_2N_3)$, the projection of the periodic orbit in the $\Sigma$-plane is plotted, with parameterization of time $\tau \in [0,2\pi\gamma(k)]$, where $2\pi\gamma(k)$ is the period of the orbit (see the end of Section \ref{sec:typeVIIandIX_Proofs}), and where the colour is determined by the value of $\sin(\tau/\gamma)$; see the colour bar. Right: plot of the period as a function of $\log_{10}(N_1N_2N_3)$, i.e. the period decreases monotonically with respect to $N_1N_2N_3$.}\label{fig:TypeVIII_continuation}
\end{figure}
%
%
\vspace{-0.8cm}
\begin{figure}[H]
\centering
\centering\includegraphics[scale=0.14]{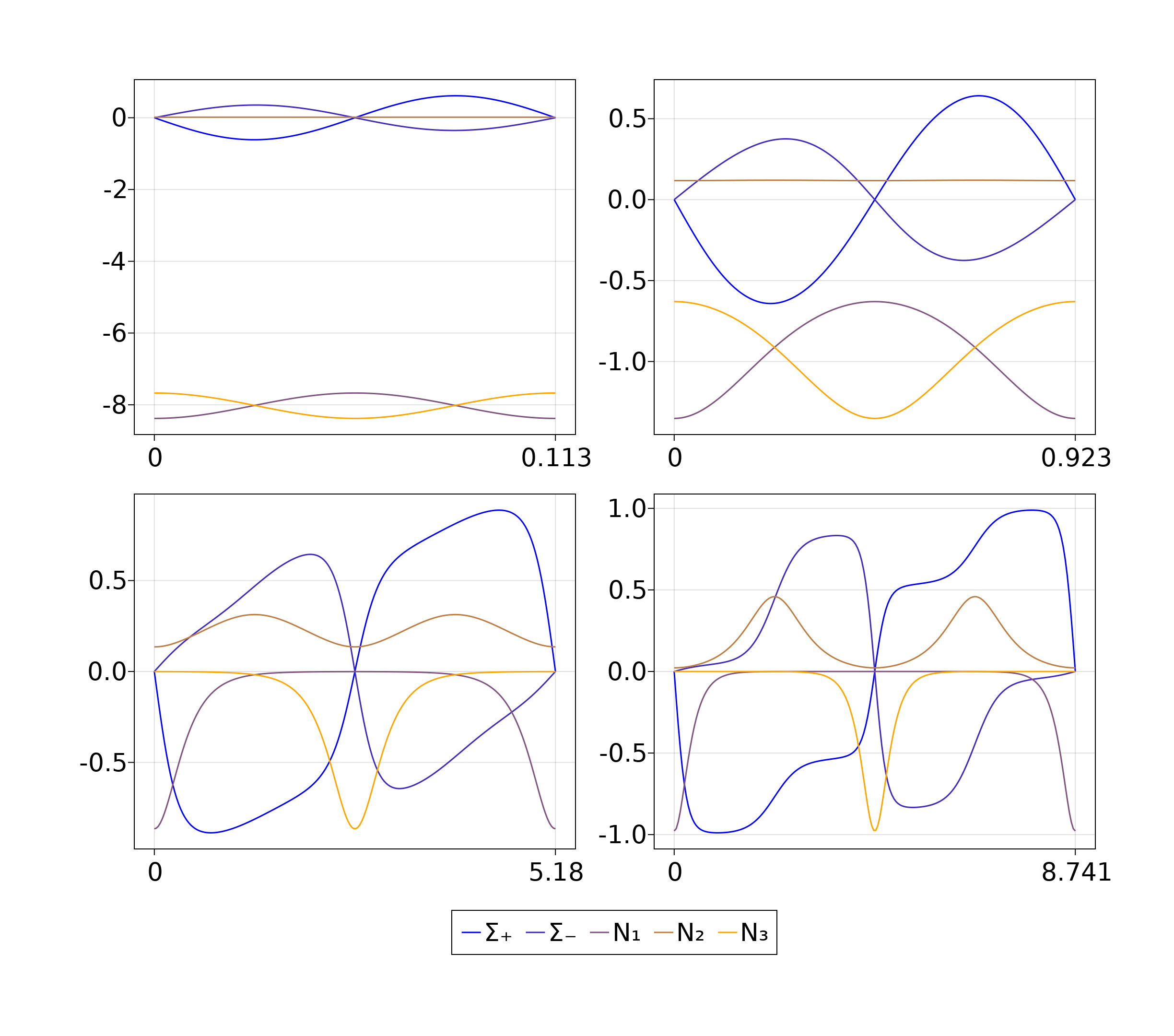}
\vspace{-0.9cm}
\caption{Time series plots of the Bianchi type $\mathrm{VIII}$ periodic orbits for some $N_1N_2N_3=\mbox{constant}$. Row 1: 
1 and $10^{-1}$ (left, right). Row 2: $10^{-4}$ and $10^{-7}$. 
One period plotted with time on the horizontal axis.}\label{fig:TypeVIII_TimeSeries}
\end{figure}

\begin{figure}[H]
\centering
\includegraphics[scale=0.15, trim = 1cm 0cm 1cm 0cm, clip]{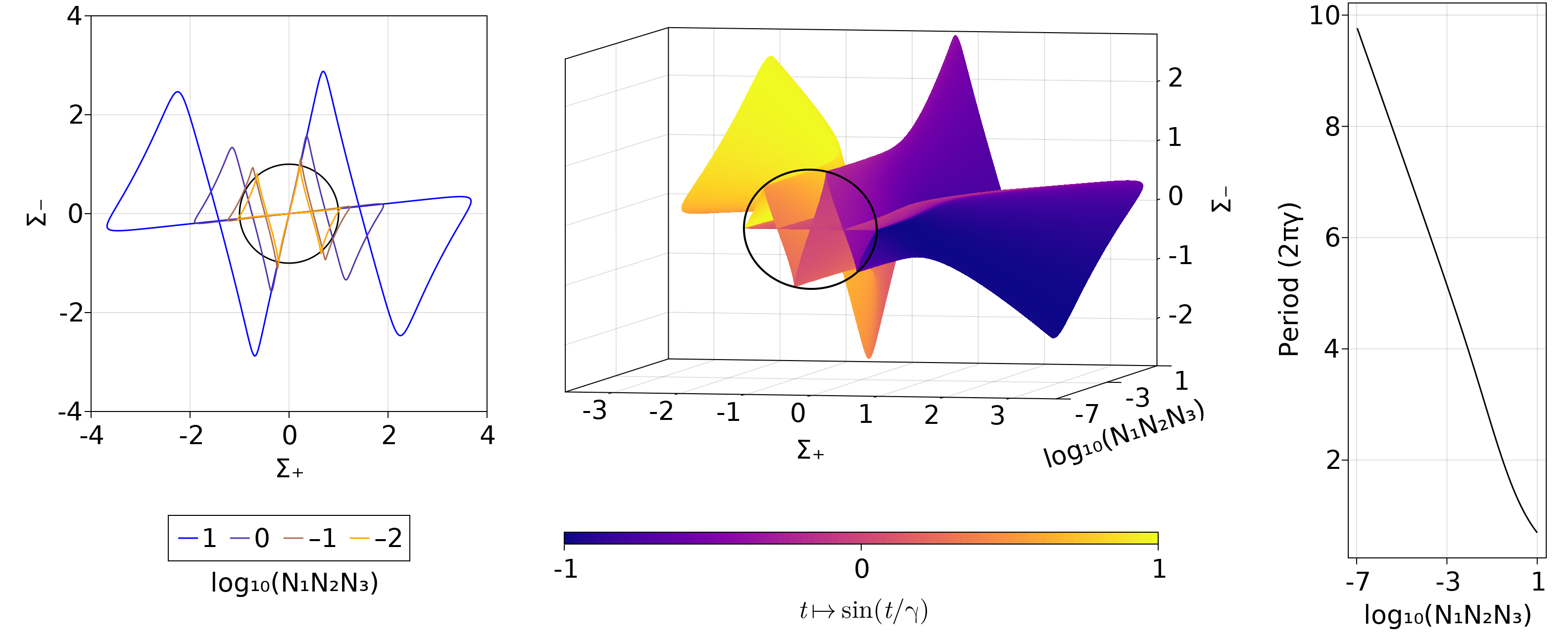}
\caption{
Left: Projection into the $\Sigma$-plane of some Bianchi type $\mathrm{IX}$ periodic orbits of Theorem \ref{thm:typeVIII_IX}. 
Even though these projections intersect the Kasner circle, note that the periodic orbits are far from the Mixmaster with a fixed `distance' $N_1N_2N_3$. 
Middle: Numerical continuation of the Bianchi type $\mathrm{IX}$ periodic orbits, 
for $N_1N_2N_3=10^k$, $k\in[-7,1]$. 
Right: plot of the period as a function of $\log_{10}(N_1N_2N_3)$.}\label{fig:TypeIX_continuation}
\end{figure}
\vspace{-0.75cm}
\begin{figure}[H]
\centering\includegraphics[scale=0.14]{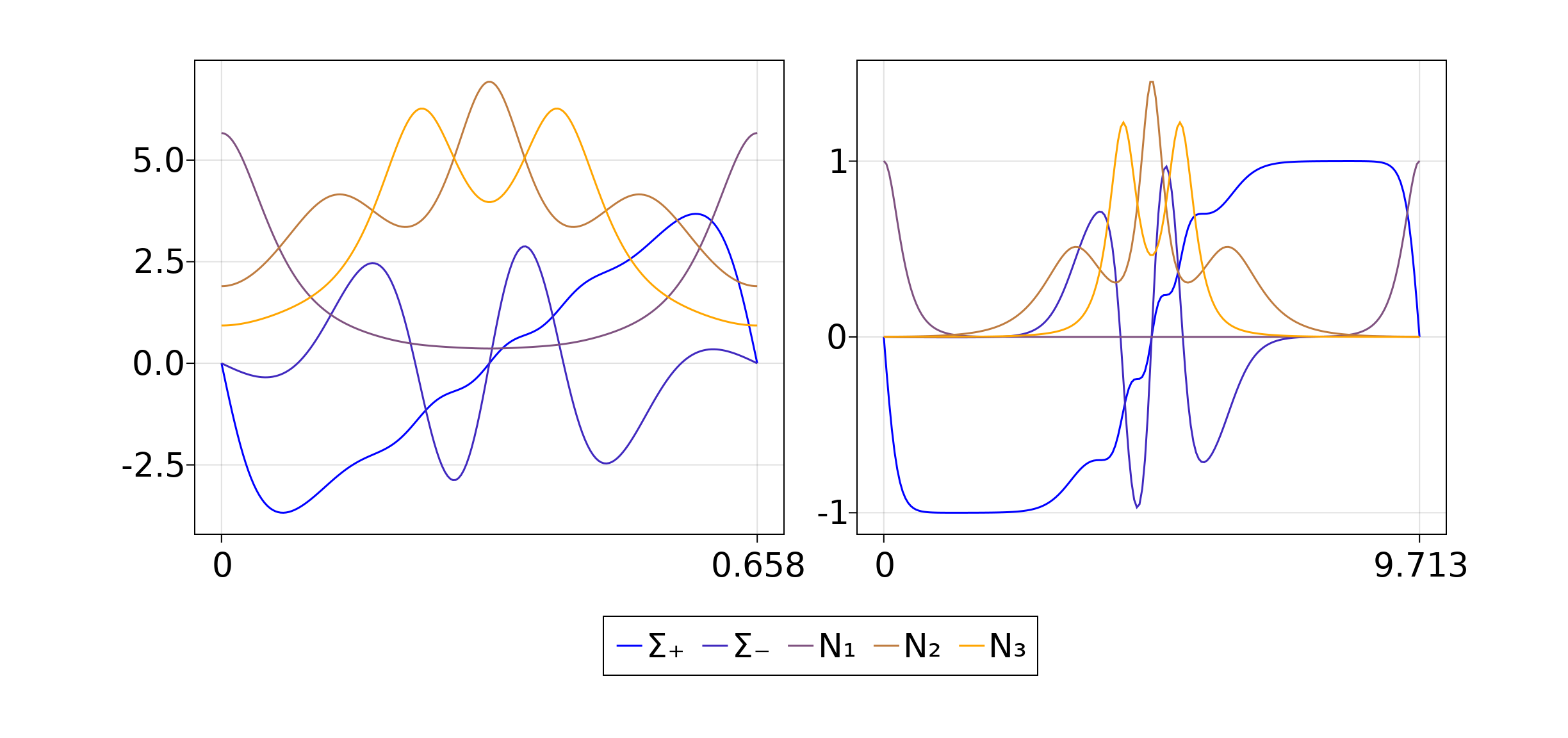}
\vspace{-0.4cm}
\caption{Time series plots of the Bianchi type $\mathrm{IX}$ periodic orbits at the two extremes: $N_1N_2N_3 = 10$ (left) and $N_1N_2N_3 = 10^{-7}$ (right). One period plotted with time on the horizontal axis.}\label{fig:TypeIX_TimeSeries}
\end{figure}
%
Numerically, the continua of periodic orbits of Theorem \ref{thm:typeVIII_IX} have boundaries ($N_1N_2N_3\rightarrow 0$) consisting of heteroclinic chains of lower Bianchi types. 
For type VIII, we conjecture that the boundary consists of the type $\mathrm{II}$ heteroclinic chain of period four that contains the Taub point $\mathrm{T}_1$, which is the limiting object of the (unique, up to $D_3$-symmetry) period four chain as $v\to 0$. 
This is supported by the case $N_1N_2N_3=10^{-7}$ in Figure~\ref{fig:TypeVIII_continuation}, its time series in Figure~\ref{fig:TypeVIII_TimeSeries} and Figure \ref{KC_0_map} (right).
For type $\mathrm{IX}$, a similar conjecture can be formulated. 
However, Figure~\ref{fig:TypeIX_TimeSeries} suggests that the boundary consists of one type $\mathrm{II}$ heteroclinic (when $N_1>0$) and one type $\mathrm{VII}_0$ heteroclinic (when $N_2>0,N_3>0$), see Figure \ref{KC_VI_VII}. 

\subsection{Computer-assisted proofs}\label{sec:typeVIIandIX_Proofs}

In this Section, we introduce the ideas behind the computer-assisted proofs (CAPs) of existence of periodic orbits. It is worth mentioning that the field of CAPs in dynamics is by now well-developed with some of the famous early pioneering works being the proof of the universality of the Feigenbaum constant \cite{feigenbaum} and the proof of existence of the strange attractor in the Lorenz system \cite{lorenz}. We refer the interested reader to the survey papers \cite{VANDENBERG_Dynamics,GOMEZ_PDESurvey,KOCH_ComputeAssisted, NAKAO_VerifiedPDE, RUMP_survey}, as well as the books \cite{MR3822720,MR3971222,TUCKER_ValidatedIntroduction}. 

In the present paper, we obtain CAPs via a Newton-Kantorovich like theorem (see \cite{MR0231218} for the original version and \cite{VANDENBERG_RigorousChaos} for more details). 
Since the vector field in \eqref{full:subsv=0} is analytic, then solutions are also analytic.
Hence, we expand a periodic solution in Fourier series and interpret its Fourier coefficients as isolated zeros of a mapping amenable for a Newton-like method. The constructed (Newton-like) fixed-point operator is contracting in the vicinity of a numerical approximation of the zero yielding a CAP via a Newton-Kantorovich argument. 
Our approach is inspired by the work of Yamamoto \cite{MR1639986}, the infinite-dimensional Krawczyk operator \cite{MR2394226} and more closely by the approach proposed in \cite{DAY_ValidatedContinuation}.
Note that functional analytic methods of CAPs for studying periodic orbits of differential equations go back to the work of Cesari on Galerkin projections for periodic solutions \cite{MR0151678, MR0173839}.

%
Recall that there are two conserved quantities for solutions of \eqref{full:subsv=0}, given by $\Sigma^2 + \Omega_k$ which is the right-hand side in \eqref{constraint} and $\Delta$ in \eqref{eq:conserved_quantity}. Note that for the Bianchi types $\mathrm{VIII}$ and $\mathrm{IX}$ with $v=0$, the conserved quantity $\Delta$ is equivalent to the product $N_1 N_2 N_3$ being constant.
This yields a two-parameter family of periodic solutions; whence, to isolate the periodic orbits, we search for $2\pi$-periodic orbits of the following auxiliary ODE
\begin{subequations}\label{eq:auxiliary_ode}
	\begin{align}
	\frac{d}{dt} \Sigma_+ &= 2\gamma \left[ \left(N_3-N_2\right)^2 - N_1\left(2N_1-N_2-N_3\right) \right]  + \eta_1 \Sigma_+, \label{eq:auxiliary_odeSigma+} \\
	\frac{d}{dt} \Sigma_- &= 2\gamma \sqrt{3} \left(N_3-N_2\right)\left(N_1-N_2-N_3\right), \label{eq:auxiliary_odeSigma-} \\
	\frac{d}{dt} N_1 	   &= 4\gamma \Sigma_+ N_1, \label{eq:auxiliary_odeN1} \\
	\frac{d}{dt} N_2 	   &= -2\gamma ( \Sigma_+ + \sqrt{3}\Sigma_-)N_2 + \eta_2, \label{eq:auxiliary_odeN2} \\
	\frac{d}{dt} N_3 	   &= -2\gamma ( \Sigma_+ - \sqrt{3}\Sigma_- )N_3, \label{eq:auxiliary_odeN3}
	\end{align}
\end{subequations}
where $\gamma, \eta_1, \eta_2 \in \mathbb{C}$ are indeterminate constants.

Whenever $\eta_1 = \eta_2 = 0$, the equations \eqref{eq:auxiliary_ode} reduce to the original system \eqref{full:subsv=0} scaled by a factor $\gamma$.
The constant $\gamma$ arises from a time scaling, $\tau \mapsto t(\tau) := \gamma^{-1}\tau$, which turns a $2\pi\gamma$-periodic solution into a $2\pi$-periodic solution. The constants $\eta_1,\eta_2$ play the role of \textit{unfolding parameters} such that $\eta_1 = \eta_2 = 0$ under suitable conditions. 
The 
next lemma, reminiscent of the approach presented in \cite{Brugos_L_MJ},  gives sufficient conditions for the periodic solutions of \eqref{eq:auxiliary_ode} to yield periodic solutions of \eqref{full:subsv=0}.
\begin{lem}\label{lem:aux_po}
Let $t \in \mathbb{R} \mapsto (\Sigma_+(t), \Sigma_-(t), N_1(t), N_2(t), N_3(t))$ be a $2\pi$-periodic solution of \eqref{eq:auxiliary_ode}.
If $\Sigma_+, N_1 N_3$ are not identically zero and $N_1 N_3$ has constant sign, then $\eta_1 = \eta_2 = 0$. \newline
Additionally, if $\gamma \in (0, +\infty)$, then $\tau \in \mathbb{R} \mapsto (\Sigma_+(\gamma^{-1}\tau), \Sigma_-(\gamma^{-1}\tau), N_1(\gamma^{-1}\tau), N_2(\gamma^{-1}\tau), N_3(\gamma^{-1}\tau))$ is a $2\pi \gamma$-periodic solution of \eqref{full:subsv=0}.
\end{lem}
\begin{proof}
Given a $2\pi$-periodic solution of \eqref{eq:auxiliary_ode}, then $N_1 N_2 N_3$ is also $2\pi$-periodic such that
\begin{equation}
0 = \int_0^{2\pi} \frac{d}{dt}\left(N_1(t) N_2(t) N_3(t)\right) \, dt = \eta_2 \int_0^{2\pi} N_1 (t) N_3 (t) \, dt.
\end{equation}
Since $N_1 N_3$ has constant sign and is not identically zero, the previous equality implies $\eta_2 = 0$.
Similarly, $\Sigma^2 + \Omega_k$ is $2\pi$-periodic satisfying
\begin{equation}\label{2ndequality}
0 = \int_0^{2\pi} \frac{d}{dt}\left(\Sigma^2 + \Omega_k\right) \, dt = 2 \eta_1\int_0^{2\pi} \Sigma_+(t)^2 \, dt,
\end{equation}
where we used $\eta_2 = 0$. Since $\Sigma_+$ is not identically zero, the equality \eqref{2ndequality} implies $\eta_1 = 0$.

Lastly, since $\gamma \in (0, +\infty)$, it is straightforward to check that the derivative of the solution $\tau \in \mathbb{R} \mapsto (\Sigma_+(\gamma^{-1}\tau), \Sigma_-(\gamma^{-1}\tau), N_1(\gamma^{-1}\tau), N_2(\gamma^{-1}\tau), N_3(\gamma^{-1}\tau))$ satisfies \eqref{full:subsv=0}.
\end{proof}

The periodic solutions of \eqref{eq:auxiliary_ode} are analytic, since they are solutions of an  ODE defined by an analytic vector field. Therefore, they admit a Fourier expansion whose Fourier coefficients decay exponentially. The Banach space of bi-infinite sequences with geometric decay\footnote{Note that $X_\textnormal{Fourier}$ is $\ell^1$ when $\nu=1$ and the sequences do not enjoy decay. In case $\nu>1$, a sequence $a \in X_\textnormal{Fourier}$ decays exponentially fast to zero with decay rate at least $\nu$.} rate $\nu \geq 1$ is defined as
\begin{equation*}
    X_\textnormal{Fourier} := \left\{ a \in \mathbb{C}^\mathbb{Z} \, : \, | a |_{X_\textnormal{Fourier}} := \sum_{k \in \mathbb{Z}} |a_k| \nu^{|k|} < +\infty \right\}.
\end{equation*}
The Banach space $X_\textnormal{Fourier}$ becomes a Banach algebra with the discrete convolution denoted by $* : X_\textnormal{Fourier} \times X_\textnormal{Fourier} \to X_\textnormal{Fourier}$ and defined by $a * b := \left\{ \sum_{l \in \mathbb{Z}} a_{k - l} b_l \right\}_{k \in \mathbb{Z}}$.
%
Let $X := X_\textnormal{Fourier}^5 \times \mathbb{C}^3$ be the Banach space equipped with the norm
\begin{equation*}
    |x|_X := \max\left( \max_{j =1, \dots, 5} |a_j|_{X_\textnormal{Fourier}}, |\gamma|, |\eta_1|, |\eta_2|\right), \quad \text{for all } x := (a_1, a_2, a_3, a_4, a_5, \gamma, \eta_1, \eta_2) \in X.
\end{equation*}
Fix $c \in \mathbb{R}$. Note that, in Theorem~\eqref{thm:typeVIII_IX}, we fix $c = 10^k$ for each $k \in \{1, 0, \dots, -7\}$. Requiring $\left(\Sigma^2 + \Omega_k\right)|_{t=0} = 1, N_1(0)N_2(0)N_3(0) = c, \Sigma_+(0)=0$ and plugging
\begin{subequations}
\begin{align*}
\Sigma_+(t) &= \sum_{k \in \mathbb{Z}} (a_1)_k e^{i k t}, \quad \Sigma_-(t)  = \sum_{k \in \mathbb{Z}} (a_2)_k e^{i k t}, \\
N_1(t)          &= \sum_{k \in \mathbb{Z}} (a_3)_k e^{i k t}, \quad N_2(t) = \sum_{k \in \mathbb{Z}} (a_4)_k e^{i k t}, \quad N_3(t) = \sum_{k \in \mathbb{Z}} (a_5)_k e^{i k t},
\end{align*}
\end{subequations}
into \eqref{eq:auxiliary_ode}, yields the unbounded operator $F : X \to X$ given by
\begin{equation*}
F(a_1, a_2, a_3, a_4, a_5, \gamma, \eta_1, \eta_2) :=
\begin{pmatrix}
2\gamma \left[ \left(a_5-a_4\right)^{*2} - a_3 * \left(2a_3-a_4-a_5\right) \right] + \eta_1 a_1 - \mathcal{D}(a_1)\\
2\gamma \sqrt{3} \left(a_5-a_4\right) * \left(a_3-a_4-a_5\right) - \mathcal{D}(a_2) \\
4\gamma a_1 * a_3 - \mathcal{D}(a_3) \\
-2\gamma ( a_1 + \sqrt{3}a_2 ) * a_4 + \eta_2 - \mathcal{D}(a_4) \\
-2\gamma ( a_1 - \sqrt{3}a_2 ) * a_5 - \mathcal{D}(a_5) \\
\mathcal{E}(a_1) \\
(\Sigma^2 + \Omega_k)(\mathcal{E}(a_1), \mathcal{E}(a_2), \mathcal{E}(a_3), \mathcal{E}(a_4), \mathcal{E}(a_5)) - 1 \\
\mathcal{E}(a_3) \mathcal{E}(a_4) \mathcal{E}(a_5) - c
\end{pmatrix},
\end{equation*}
where $a^{*2} := a*a$, $\mathcal{D}(a) = \{i k a_k\}_{k \in \mathbb{Z}}$ for all $a \in X_\textnormal{Fourier}$ such that $\mathcal{D}(a) \in X_\textnormal{Fourier}$ and $\mathcal{E}(a) = \sum_{k \in \mathbb{Z}} a_k$ for all $a \in X_\textnormal{Fourier}$. Formally, $\mathcal{D}:X_\textnormal{Fourier}\rightarrow X_\textnormal{Fourier}$ is the representation of the (unbounded) differentiation operator, and $\mathcal{E}:X_\textnormal{Fourier}\rightarrow\mathbb{C}$ is the evaluation at zero functional. Note that we introduced the requirement that $\Sigma_+(0)=0$ to quotient out the temporal translation invariance of the periodic orbit; this specific choice was motivated by numerical observations and symmetry of the system \eqref{full:subsv=0}.

The CAP consists in showing the existence of $\tilde{x} := (\tilde{a}_1, \tilde{a}_2, \tilde{a}_3, \tilde{a}_4, \tilde{a}_5, \tilde{\gamma}, \tilde{\eta}_1, \tilde{\eta}_2) \in X$ such that

\begin{enumerate}
\item[(i)] $F(\tilde x) = 0$ and $\sum_{k \in \mathbb{Z}} (\tilde{a}_j)_k e^{i k t} \in \mathbb{R}$ for $j=1,\dots, 5$ and all $t \in [0, 2\pi)$. This implies that $\tilde{a}_1, \tilde{a}_2, \tilde{a}_3, \tilde{a}_4, \tilde{a}_5$ are the Fourier coefficients of a real 
$2\pi$-periodic orbit of \eqref{eq:auxiliary_ode} given by $\left(\Sigma_+ , \Sigma_- , N_1 , N_2 , N_3 \right) :=( \sum_{k \in \mathbb{Z}} \tilde{a}_1 e^{i k t}, \sum_{k \in \mathbb{Z}} \tilde{a}_2 e^{i k t}, \sum_{k \in \mathbb{Z}} \tilde{a}_3 e^{i k t}, \sum_{k \in \mathbb{Z}} \tilde{a}_4 e^{i k t},\sum_{k \in \mathbb{Z}} \tilde{a}_5 e^{i k t} )$.
\item[(ii)] $\Sigma_+, N_1 N_3$ are not identically zero and $N_1 N_3$ has constant sign for some $\gamma \in (0, +\infty)$. Thus $\tau \in \mathbb{R} \mapsto (\Sigma_+(\gamma^{-1}\tau), \Sigma_-(\gamma^{-1}\tau), N_1(\gamma^{-1}\tau), N_2(\gamma^{-1}\tau), N_3(\gamma^{-1}\tau))$ is a real $2\pi \gamma$-periodic solution of \eqref{full:subsv=0} according to Lemma~\ref{lem:aux_po}.
\end{enumerate}

To start with, one finds a numerical approximation of a periodic orbit of \eqref{full:subsv=0} with approximate frequency $\gamma_0^{-1} > 0$; this may be achieved via a combination of an iterative procedure (e.g.\ Newton's method or gradient descent) and numerical integration of the vector fields. 
From the scaling $\tau \mapsto \gamma_0^{-1} \tau$, one obtains $2\pi$-periodic functions whose Fourier coefficients are denoted by $a_{0,1}, a_{0,2}, a_{0,3}, a_{0,4}, a_{0,5}$. Consequently, one has successfully produced a numerical periodic solution of \eqref{eq:auxiliary_ode} with parameters $\gamma = \gamma_0$ and $\eta_1 = \eta_2 = 0$. In other words, this process yields an approximate zero $x_0 = (a_{0,1}, a_{0,2}, a_{0,3}, a_{0,4}, a_{0,5}, \gamma_0, \eta_1, \eta_2)$ of $F$.

The proof of (i) relies on the contraction of a fixed-point operator in a ball centred at $x_0$. Notably, the radius $r_0$ of this ball is an a posteriori error bound for the numerical approximation $x_0$. Whence, (ii) amounts to a simple rigorous evaluation on $[0, 2\pi]$ of Fourier series with a known error bound $r_0$. The actual procedure for both (i) and (ii) are purely technical and are given in Appendix \ref{appendix:step_2}.

\section{Conclusion}\label{sec:disc}


We have described the full stratification of invariant sets within the dynamical phase-space of Bianchi models in HL gravity in case of a cubic dominant potential in \eqref{calV}. Similar to GR, the Bianchi type $\mathrm{I}$ consists of a circle of equilibria and type $\mathrm{II}$ consists of heteroclinic orbits. However, the type $\mathrm{II}$ dynamics induce the Kasner map which is not chaotic, in contrast to GR.
This is in agreement with previous results in the literature which concludes that the dynamics towards the singularity is oscillating, but not chaotic, see \cite{Kamenshchik,Bakas10}. 

Moreover, we have proved the existence of several periodic orbits of Bianchi type $\mathrm{VIII}$ and $\mathrm{IX}$ in Theorem \ref{thm:typeVIII_IX}, some of which are far from the Mixmaster attractor consisting of Bianchi type $\mathrm{I}$ and $\mathrm{II}$.  
This yields a behaviour which is not described by the BKL picture of bouncing Kasner-like states. Hence, our present results indicate that the asymptotic dynamics for $v=0$ is different from GR, and in particular, the global attractor of the ODE \eqref{full:subsv=0} is bigger than the usual Mixmaster. This shows that the Kasner map induced by the type $\mathrm{II}$ solutions does not accurately approximate the asymptotic dynamics of \eqref{full:subsv=0}, in contrast to \cite{Kamenshchik,Bakas10}.
Note that this is expected, since the function $\Delta$ in \eqref{eq:conserved_quantity}, which is a Lyapunov function that accounts for the convergence towards the Mixmaster for $v>0$, becomes a conserved quantity satisfying $\Delta'=0$ for $v=0$.

Several questions posed in Section \ref{sec:typeVIIIandIX} for $v=0$ remain to be answered.
For example, do the periodic orbits always occur in one-parameter families? 
What is the boundary of such one-parameter families? 
Are there heteroclinic connections between periodic orbits for each constant value $\Delta\in \mathbb{R}_+$? 
More generally, what is the overall asymptotic dynamics for each fixed $\Delta\in \mathbb{R}_+$, e.g. $\Delta\equiv 1$? Is the qualitative dynamics for small $\Delta$ (i.e. close to the Mixmaster) similar to the one of large $\Delta$ (i.e. far from the Mixmaster)? What is the global attractor of the ODE \eqref{full:subsv=0} bound to the constraint \eqref{constraint}?

Beyond $v=0$, the relationship of the periodic orbits in Theorem \ref{thm:typeVIII_IX} with the subcritical case, $v\in (0,1/2)$, the critical GR case, $v=1/2$, and similar models remains a mystery. Are the center-stable manifolds (of a given one-parameter family of periodic orbits for $v=0$) the limiting object of an appropriate stable manifold (of a type $\mathrm{II}$ heteroclinic chain for $v>0$), as $v\to 0$? See \cite[Conjectures 7.1-7.3]{HellLappicyUggla}.
%
Aside from that, note that the (three sets of) type $\mathrm{II}$ parallel heteroclinics orbits in HL described in Section \ref{sec:typeIandII} bear some similarity to the (two sets of) Bianchi type $\mathrm{VI}_{-1/9}$ frame transitions in Iwasawa frame. Can the ODE \eqref{full:subsv=0} shine a light on the conjecture of an attractor for the exceptional Bianchi models of class B? See \cite{Exceptional,ugg13a,ugg13b}, and in particular, the existence of periodic orbits for non-vacuum Bianchi type $\mathrm{VI}_{-1/9}$ models in \cite[Theorem 4.1]{ColeyEtal}.
%
Lastly, the model \eqref{full:subsv=0} also share some resemblance to a generalized Toda problem in two dimensions and to the H\'enon-Heiles system, see \cite[Appendix A]{HellLappicyUggla}.
Thus the ODE \eqref{full:subsv=0} provides a simple toy model that poses new directions which may guide the search for new conclusions as regards GR and other problems. 
%

\appendix

\section{Derivation of the ODE model}\label{app}
We deduce the evolution equations \eqref{full:subs} from the action \eqref{action} following \cite[Appendix A]{HellLappicyUggla}.
For the vacuum HL class~A Bianchi models, the action~\eqref{action}
expressed in terms of a symmetry adapted spatial (left-invariant) co-frame $\{{\omega}^1,{\omega}^2,{\omega}^3\}$  yields the field equations for the associated metric~\eqref{genmetric}.
Expressing the components of the spatial metric in such a symmetry adapted
spatial co-frame leads to that they become purely time-dependent in diagonal form, see~\cite{waiell97} and references therein. 
Setting the shift vector $N_i$ in~\eqref{genmetric} to zero, 
the diagonalized vacuum spatially homogeneous class~A metrics are given by
\begin{equation}\label{threemetric}
\mathbf{g} = -N^2(t)dt\otimes
dt + g_{11}(t)\:{\omega}^1\otimes {\omega}^1 +
g_{22}(t)\:{\omega}^2\otimes {\omega}^2 +
g_{33}(t)\:{\omega}^3\otimes {\omega}^3,
\end{equation}
where the lapse $N(t)$ is a non-zero function determining
the particular choice of time variable.

In order to obtain simple Hamiltonian equations, we first focus on the kinetic
part ${\cal T}$ in equation~\eqref{kin}, which can be written as
\begin{equation}\label{kin2}
{\cal T} = (K^1\!_1)^2 + (K^2\!_2)^2 + (K^3\!_3)^2 - \lambda (K^1\!_1 + K^2\!_2 + K^3\!_3)^2,
\end{equation}
where the extrinsic curvature is given by
$(K_{11}, K_{22}, K_{33}) = (\dot{g}_{11}, \dot{g}_{22}, \dot{g}_{33})/(2N)$ such that $\dot{}$ denotes a derivative with respect to $t$, and thus raising one of the indices, we obtain that $(K^1\!_1, K^2\!_2, K^3\!_3) 
= \left(\dot{g}_{11}/g_{11}, \dot{g}_{22}/g_{22}, \dot{g}_{33}/g_{33}\right)/(2N)$.
%

To simplify ${\cal T}$, we make a variable transformation
from the metric components to the variables $\beta^0,\beta^+,\beta^-$,
first introduced by Misner~\cite{mis69a,mis69b,grav73},
\begin{equation}\label{Misnerbeta}
g_{11} = e^{2( \beta^0 - 2\beta^+)}, \qquad
g_{22} = e^{2(\beta^0 + \beta^+ + \sqrt{3}\beta^-)},\qquad
g_{33} = e^{2(\beta^0 + \beta^+ - \sqrt{3}\beta^-)}.
\end{equation}
This results in that ${\cal T}$ in equation~\eqref{kin2} takes the form
\begin{equation}\label{calT}
{\cal T} = \dfrac{6}{N^2}
\left[-\left(\frac{3\lambda-1}{2}\right)(\dot{\beta}^0)^2 + (\dot{\beta}^+)^2 + (\dot{\beta}^-)^2\right].
\end{equation}
Note that the character of the quadratic form~\eqref{calT} changes
when $\lambda = 1/3$. Since we are interested
in continuously deforming the GR case $\lambda = 1$, we restrict considerations to
$\lambda > 1/3$. To simplify the kinetic part further, we introduce a new variable
$\beta^\lambda$ and a density-normalized lapse function $\mathcal{N}$, defined by
\begin{equation}\label{betaLAMBDAandNcal}
\beta^\lambda := \sqrt{\frac{3\lambda-1}{2}}\beta^0,\qquad
{\cal N} := \frac{N}{12\sqrt{g}},
\end{equation}
where $g = g_{11}g_{22}g_{33} = \exp(6\beta^0)$ is the determinant of the spatial metric
in the symmetry adapted co-frame, which leads to,
\begin{equation}
\sqrt{g}N{\cal T} = \dfrac{1}{2 {\cal N}}\left[-(\dot{\beta}^\lambda)^2 + (\dot{\beta}^+)^2 + (\dot{\beta}^-)^2\right].
\end{equation}

It is convenient to define $T := \frac{\sqrt{g} N}{\mathcal{N}}\mathcal{T} = 12g{\cal T}$, so that ${\cal N}T$ is the kinetic part of the Lagrangian for
the present spatially homogeneous models, in
analogy with the GR case, see e.g., ch. 10 in~\cite{waiell97}.
The density-normalized lapse ${\cal N}$ is kept in the kinetic term
${\cal N}T$, since it is needed in order to obtain the Hamiltonian
constraint, which is accomplished by varying ${\cal N}$ in the Hamiltonian.

To proceed to a Hamiltonian description, we introduce the canonical momenta
\begin{equation}
p_\lambda := -\frac{\dot{\beta}^\lambda}{{\cal N}}, \qquad p_\pm :=\frac{\dot{{\beta}}^\pm}{{\cal N}}.
\end{equation}
This leads to that $T$ takes the form
\begin{equation}\label{Tkinetics}
T = \frac12\left(- p_\lambda^2 + p_+^2 + p_-^2\right).
\end{equation}

Similarly to the treatment of the kinetic part, we define
\begin{equation}
V := \sqrt{g}N {\cal V}/{\cal N}=12 g {\cal V}.
\end{equation}
Due to~\eqref{calV},
\begin{equation}\label{VHL}
V = {}^1V + {}^2V + {}^3V + {}^4V + {}^5V + {}^6V + \dots ,
\end{equation}
where
\begin{subequations}\label{pots}
\begin{alignat}{3}
{}^1V &:= 12k_1 gR, &\qquad {}^2V &:= 12k_2 gR^2,  &\qquad
{}^3V &:= 12k_3 g R^i\!_jR^j\!_i,\\
{}^4V &:= 12k_4 g R^i\!_jC^j\!_i,
&\qquad {}^5V &:= 12k_5 g C^i\!_jC^j\!_i, &\qquad {}^6V &:= 12k_6 gR^3.
\end{alignat}
\end{subequations}
The superscripts on ${}^AV$ (where $A = 1,\dots,6$) thereby coincide with the
subscripts of the constants $k_A$ in~\eqref{calV}.

Based on~\eqref{action}, this leads to a Hamiltonian $H$ given by
\begin{equation}\label{LambdaRham}
H := \sqrt{g}N({\cal T} + {\cal V}) = {\cal N}(T + V) = 0,
\end{equation}
where $T$ only depends on the canonical momenta $p_\lambda$, $p_\pm$,
given by~\eqref{Tkinetics}, and $V$ only depends on
$\beta^\lambda$, $\beta^\pm$, given by~\eqref{VHL} and~\eqref{pots}.

In order to derive the ordinary differential equations for these models
via the Hamiltonian equations in terms of the variables
$\beta^\lambda$, $\beta^\pm$ and the canonical momenta $p_\lambda$, $p_\pm$,
we need to compute each ${}^AV(\beta^\lambda,\beta^\pm)$.
We proceed with the simplest case that minimally modifies vacuum GR in the present context, the vacuum $\lambda$-$R$ models~\cite{giukie94,belres12,lolpir14}.
They are obtained from an action that consists of the generalized
kinetic part in~\eqref{kin}, i.e, by keeping $\lambda$ (GR is obtained by
setting $\lambda=1$), and the vacuum GR potential in~\eqref{calV},
i.e., a potential arising from $-R$ only, and hence when $k_1=-1$
and $k_2=k_3=k_4=k_5=k_6=0$ in~\eqref{calV}. These models suffice for
our goal of deriving the ODEs \eqref{full:subs}.
The case that modifies GR with more general potentials, the HL models are similar and can be found in~\cite[Appendix A.2]{HellLappicyUggla}. In particular, they heuristically argue that a broad class of HL models possess a dominant potential with asymptotic dynamics described by the $\lambda$-$R$ models.

To obtain succinct expressions for the spatial curvature, and thereby the potential $V={}^1V = - 12gR$,
we introduce the following auxiliary quantities 
\begin{subequations}\label{malpha}
\begin{align}
m_1 &:= n_1g_{11} = n_1 e^{2(2v\beta^\lambda - 2\beta^+)},\\
m_2 &:= n_2g_{22} = n_2 e^{2(2v\beta^\lambda + \beta^+ + \sqrt{3}\beta^-)},\\
m_3 &:= n_3g_{33} = n_3 e^{2(2v\beta^\lambda + \beta^+ - \sqrt{3}\beta^-)}.
\end{align}
\end{subequations}
Here we have introduced the parameter $v$, which is defined by
the relation
\begin{equation}\label{app:v}
v := \frac{1}{\sqrt{2(3\lambda - 1)}},
\end{equation}
and hence $\beta^0 = 2v\beta^\lambda$ due to~\eqref{betaLAMBDAandNcal}.
The parameter $v$ plays a prominent role in the evolution equations. Since we are
interested in continuous deformations of GR with $\lambda=1$, and thus
$v=1/2$, we restrict attention to $v\in (0,1)$.
Specializing the general expression for the spatial curvature in~\cite{elsugg97}
to the diagonal class A Bianchi models leads to 
\begin{equation}\label{R11}
R^1\!_1 = \frac{1}{2g}(m_1^2 - (m_2-m_3)^2),
\end{equation}
where $R^1\!_1 = g^{11}R_{11}
= g_{11}^{-1}R_{11}$, and similarly by permutations
for $R^2\!_2$ and $R^3\!_3$. It follows that the spatial scalar curvature
$R = R^1\!_1 + R^2\!_2 +R^3\!_3$ is given by
\begin{equation}\label{Rscalar}
R = -\frac{1}{2g}(m_1^2 + m_2^2 + m_3^2 - 2m_1m_2 - 2m_2m_3 - 2m_3m_1).
\end{equation}
This thereby yields the potential in~\eqref{VHL} and~\eqref{pots} with $k_1=-1$:
\begin{equation}\label{kinV}
V = {}^1V = -12 g R = 6(m_1^2 + m_2^2 + m_3^2 - 2m_1m_2 - 2m_2m_3 - 2m_3m_1),
\end{equation}
where $V$ depends on $\beta^\lambda$ and $\beta^\pm$ via $m_1$, $m_2$ and $m_3$,
according to equation~\eqref{malpha}.

The evolution equations for $\beta^\lambda$, $\beta^\pm$, $p_\lambda$, $p_\pm$
are obtained from Hamilton's equations, where $T$ and $V$
in the Hamiltonian~\eqref{LambdaRham} are given by~\eqref{Tkinetics} and~\eqref{kinV},
respectively, which yields
\begin{subequations}\label{HamiltonEQ}
\begin{align}
\dot{\beta}^\lambda &= \frac{\partial H}{\partial p_\lambda}
= -{\cal N} p_\lambda, \qquad && \dot{p}_\lambda
= - \frac{\partial H}{\partial\beta^\lambda}
= - \mathcal{N} \frac{\partial V}{\partial\beta^\lambda}, \label{betadotham}\\
\dot{\beta}^\pm &= \frac{\partial H}{\partial{p}_\pm}
= {\cal N}{p}_\pm, \qquad &&\dot{p}_\pm
= - \frac{\partial H}{\partial{\beta}^\pm}
= - \mathcal{N} \frac{\partial V}{\partial{\beta}^\pm},
\end{align}
\end{subequations}
while the Hamiltonian constraint $T+V=0$ is obtained by varying ${\cal N}$.

Next, we choose a new time variable
$\tau_-:=-\beta^\lambda$, 
which is directed toward the physical past, since we are considering
expanding models. This is accomplished by setting ${\cal N} =  p_\lambda^{-1}$
in the first equation in~\eqref{betadotham}, and thereby $N = 12\sqrt{g}/ p_\lambda$,
which results in the following evolution equations:
\begin{subequations}\label{HamiltonEQ2}
\begin{align}
\frac{d\beta^\lambda}{d\tau_-} &= -1, \qquad
&&\frac{d p_\lambda}{d\tau_-} =
-\frac{1}{ p_\lambda} \frac{\partial V}{\partial\beta^\lambda}, \\
\frac{d\beta^\pm}{d\tau_-} &= \frac{{p}_\pm}{ p_\lambda}, \qquad &&
\frac{dp_\pm}{d\tau_-} = -\frac{1}{ p_\lambda} \frac{\partial V}{\partial{\beta}^\pm}.
\end{align}
\end{subequations}

We then rewrite the system~\eqref{HamiltonEQ2} and the constraint $T+V=0$
using the non-canonical variable transformation,
\begin{equation}\label{SigmaNvariables}
\Sigma_\pm := - \frac{p_\pm}{ p_\lambda}, \qquad \qquad \qquad N_\alpha
:= - 2\sqrt{3}\left(\frac{m_\alpha}{ p_\lambda}\right),
\end{equation}
while keeping $p_\lambda$. Note that $\Sigma_\pm = d\beta^\pm/d\beta^\lambda = -d\beta^\pm/d\tau_-$.

These variables lead to a decoupling 
of the evolution equation for the variable $ p_\lambda$,
\begin{equation}\label{p0prime}
p_\lambda^\prime = -4v(1-\Sigma^2) p_\lambda,
\end{equation}
where ${}^\prime$ denotes the derivative $d/d\tau_-$.
This yields the following reduced system of evolution equations
\begin{subequations}\label{dynsyslambdaR}
\begin{align}
\Sigma_\pm^\prime &= 4v(1-\Sigma^2)\Sigma_\pm + {\cal S}_\pm,\\
N_1^\prime &= -2(2v\Sigma^2 - 2\Sigma_+)N_1,\\
N_2^\prime &= -2(2v\Sigma^2 + \Sigma_+ + \sqrt{3}\Sigma_-)N_2,\\
N_3^\prime &= -2(2v\Sigma^2 + \Sigma_+ - \sqrt{3}\Sigma_-)N_3,
\end{align}
while the Hamiltonian constraint $T+V=0$ results in
\begin{equation}\label{constrpmVIIIIX}
1 - \Sigma^2 - \Omega_k =0,
\end{equation}
\end{subequations}
where
\begin{subequations}\label{LambdaRquantities}
\begin{align}
\Sigma^2 &:= \Sigma_+^2 + \Sigma_-^2,\\
\Omega_k &:= N_1^2 + N_2^2 + N_3^2 - 2N_1N_2 - 2N_2N_3 - 2N_3N_1,\\
{\cal S}_+ &:= 2[(N_2 - N_3)^2 - N_1(2N_1 - N_2 - N_3)],\\
{\cal S}_- &:= 2\sqrt{3}(N_2 - N_3)(N_2 + N_3 - N_1).
\end{align}
\end{subequations}

Note that the variables $\Sigma_\pm$, $N_1$, $N_2$ and $N_3$,
defined in~\eqref{SigmaNvariables}, are \emph{dimensionless}. Dimensions
can be introduced in various ways, but terms in a sum
must all have the same dimension. The constraint~\eqref{constrpmVIIIIX}
is such a sum. Since this sum contains 1,
which 
obviously is dimensionless, it follows that $\Sigma_+$, $\Sigma_-$,
$N_1$, $N_2$ and $N_3$ are dimensionless, and so is the time variable
$\tau_-$, as follows from inspection of~\eqref{dynsyslambdaR}.
The vacuum GR equations are obtained by setting $v=1/2$.

In~\cite[Appendix A.2]{HellLappicyUggla}, it is heuristically argued that a broad range of HL models have asymptotic dynamics described by the $\lambda$-$R$ evolution equations \eqref{dynsyslambdaR}. To achieve this, they use Misner's approximation scheme of a `particle' moving in a potential well in $(\beta^+,\beta^-)\in\mathbb{R}^2$ space as $\tau_- = - \beta^\lambda \rightarrow \infty$, which was introduced to understand the initial Bianchi type $\mathrm{IX}$ singularity in GR, see~\cite{mis69a,mis69b,waiell97,jan01}.
For HL, each potential term in \eqref{calV} has its associated `moving walls' that move with velocity ${}^iv$. Among those, there is a dominant potential term which yields the same evolution equations as the $\lambda$-$R$ models in \eqref{dynsyslambdaR}, but with different parameters ${}^iv$ given by \eqref{v's} instead of the parameter $v$ in \eqref{app:v}.

\section{Computer-assisted proof}
\label{appendix:step_2}

We now proceed to prove the two remaining claims (i) and (ii) in Section \ref{sec:typeVIIandIX_Proofs}.

The following proposition is the core result to complete (i). More precisely, given a numerical approximation $x_0$ of a zero of $F$, it gives sufficient conditions to find the radius $r_0$ of a ball centred at $x_0$ within which there exists a unique true zero $\tilde{x}$ of $F$.

\begin{prop}\label{prop:radii_polynomial}
Let $R > 0$ and $x_0 := (a_{0,1}, a_{0,2}, a_{0,3}, a_{0,4}, a_{0,5}, \gamma_0, 0, 0) \in X$. Consider a linear bounded injective operator $\mathcal{A} : X \to X$ satisfying
\begin{subequations}
\begin{align}
|\mathcal{A} F(x_0)|_X &\leq Y, \\
| I - \mathcal{A} DF(x_0) |_{\mathscr{B}(X, X)} &\leq  Z_1, \\
\sup_{x \in \textnormal{cl}(B_R(x_0))} | \mathcal{A} D^2 F(x) |_{\mathscr{B}(X^2, X)}  &\leq Z_2.
\end{align}
\end{subequations}
for some constants $Y, Z_1, Z_2 \geq 0$. If there exists $r_0 \in [0, R]$ such that
\begin{equation}\label{prop4.3:bounds}
Y + (Z_1 - 1) r_0 + \frac{Z_2}{2} r_0^2 \leq 0 \qquad \text{and} \qquad Z_1 + Z_2 r_0 < 1,  
\end{equation}
then there exists a unique $\tilde{x} := (\tilde{a}_1, \tilde{a}_2, \tilde{a}_3, \tilde{a}_4, \tilde{a}_5, \tilde{\gamma}, \tilde{\eta}_1, \tilde{\eta}_2) \in \textnormal{cl}(B_{r_0}(x_0))$ such that $F(\tilde{x}) = 0$.

Furthermore, if $\gamma_0 \in \mathbb{R}$ and $a_{0,1}, a_{0,2}, a_{0,3}, a_{0,4}, a_{0,5}$ are sequences of Fourier coefficients of real Fourier series, then $\tilde{\gamma}, \tilde{\eta}_1, \tilde{\eta}_2 \in \mathbb{R}$ and $\tilde{a}_1, \tilde{a}_2, \tilde{a}_3, \tilde{a}_4, \tilde{a}_5$ are sequences of Fourier coefficients of real Fourier series.
\end{prop}

\begin{proof}
The proof is constructive to provide practical computational insights. To summarize, we construct the aforementioned operator $\mathcal{A}$ and the bounds $Y, Z_1, Z_2$. Secondly, we obtain a Newton-like operator $T$ which is a contraction in $\textnormal{cl}(B_{r_0}(x_0))$, given that the error bound $r_0\in [0,R]$ satisfies \eqref{prop4.3:bounds}. This yields a zero $\tilde{x}$ of $F$. Lastly, we address the properties of $\tilde{x}$ given in the last paragraph of the proposition.

Firstly, given a fixed projection dimension number $n \in \mathbb{N}$, consider the projection operator $\pi^n : X_\textnormal{Fourier} \to X_\textnormal{Fourier}$ defined by
\begin{equation}\label{eq:proj}
(\pi^n a)_k :=
\begin{cases}
a_k, & |k| \leq n,\\
0, & |k| \geq n+1,
\end{cases} \qquad \text{for all } a \in X_\textnormal{Fourier}.
\end{equation}
This operator is extended to an operator (denoted with the same symbol) on $X$ by
\begin{equation}\label{eq:proj2}
\pi^n x = (\pi^n a_1, \pi^n a_2, \pi^n a_3, \pi^n a_4, \pi^n a_5, \gamma, \eta_1, \eta_2), \quad \text{for all } x := (a_1, a_2, a_3, a_4, a_5, \gamma, \eta_1, \eta_2) \in X.
\end{equation}
It is clear that $\pi^n$ is a projection and, defining $\pi^{\infty(n)} := I - \pi^n$, we have the decomposition $X = \pi^n X \oplus \pi^{\infty(n)} X$. Intuitively, to obtain norms estimate in $X$, we carefully split the bounds into a part in $\pi^n X$ handled by the computer and a part in $\pi^{\infty(n)} X$ controlled theoretically.

Introduce a Banach space $\mathcal{X} := \mathcal{X}_\textnormal{Fourier}^5\times\mathbb{C}^3$, where 
\begin{equation*}
    \mathcal{X}_\textnormal{Fourier}:= \left\{ a\in\mathbb{C}^\mathbb{Z} : |a|_{\mathcal{X}_\textnormal{Fourier}}:=\sum_{k\in\mathbb{Z}}|a_k|\frac{\nu^{|k|}}{|k|}\right\}.    
\end{equation*}

It is clear that the unbounded operator $F$ can be interpreted as a bounded nonlinear operator $F:X \to \mathcal{X}$. Similarly, consider the bounded linear operator $L : X \to \mathcal{X}$ given by
\begin{equation*}
L x :=
\begin{pmatrix}
-\mathcal{D}(a_1) \\
-\mathcal{D}(a_2) \\
-\mathcal{D}(a_3) \\
-\mathcal{D}(a_4) \\
-\mathcal{D}(a_5) \\
\mathcal{E}(a_1) \\
\sum_{j = 1}^5 \left[\partial_{a_j} (\Sigma^2 + \Omega_k)(a_1, \dots, a_5)\right]_{a_1 = a_{0,1}, \dots, a_5 = a_{0,5}} a_j \\
\sum_{j = 3}^5 \left[\partial_{a_j} (\mathcal{E}(a_3) \mathcal{E}(a_4), \mathcal{E}(a_5))\right]_{a_3 = a_{0,3},a_4 = a_{0,4},a_5 = a_{0,5}} a_j
\end{pmatrix},
\end{equation*}
for all $x := (a_1, a_2, a_3, a_4, a_5, \gamma, \eta_1, \eta_2) \in X$ such that $Lx \in X$. 

In this scope, let $L|_{\pi^{\infty(2n)} X} : \pi^{\infty(2n)} X \to \pi^{\infty(2n)} X$ be defined as $L|_{\pi^{\infty(2n)} X} x = \pi^{\infty(2n)} L x$ for all $x \in \pi^{\infty(2n)} X$ such that $Lx \in X$. Also, set $\mathcal{A} := A \pi^{2n} + L|_{\pi^{\infty(2n)} X}^{-1} \pi^{\infty(2n)}$ where $A : \pi^{2n} X \to \pi^{2n} X$ is defined as an approximation of $(\pi^{2n} DF(x_0) \pi^{2n})^{-1}$. By construction, the operator $\mathcal{A}:\mathcal{X}\rightarrow X$ is linear, bounded and injective.

Let us give explicit formulae for the bounds $Y, Z_1, Z_2$:
\begin{enumerate}
\item
For $a, b \in \pi^n X_{\textnormal{Fourier}}$, we have $a * b \in \pi^{2n} X_{\textnormal{Fourier}}$. Hence, $F(x_0) \in \pi^{2n} X$ and
\begin{equation*}
|\mathcal{A} F(u_0)|_X = |A F(u_0)|_X =: Y.    
\end{equation*}
\item
For $a \in \pi^n X_{\textnormal{Fourier}}$ and $\mathcal{M}_a (b) := a * b$ for all $b \in X_\textnormal{Fourier}$, we have $\pi^{2n} \mathcal{M}_a = \pi^{2n} \mathcal{M}_a \pi^{3n}$. Hence, $\pi^{2n} DF(x_0) = \pi^{2n} DF(x_0) \pi^{3n} + \pi^{2n} L \pi^{\infty(3n)}$ and we obtain
\begin{align*}
&| I - \mathcal{A} DF(x_0)|_{\mathscr{B}(X, X)} \\
&= | I - (A \pi^{2n} + L|_{\pi^{\infty(2n)}X}^{-1} \pi^{\infty(2n)}) DF(x_0) |_{\mathscr{B}(X, X)} \\
&= | I - A \pi^{2n} DF(x_0) + L|_{\pi^{\infty(2n)}X}^{-1} \pi^{\infty(2n)} L + L|_{\pi^{\infty(2n)}X}^{-1} \pi^{\infty(2n)} ( DF(x_0) - L ) |_{\mathscr{B}(X, X)} \\
&\leq | \pi^{2n} - A \pi^{2n} DF(x_0) \pi^{3n} |_{\mathscr{B}(X, X)} + | L|_{\pi^{\infty(2n)}X}^{-1} |_{\mathscr{B}(X, X)} | DF(x_0) - L |_{\mathscr{B}(X, X)} + \\ &\qquad | A |_{\mathscr{B}(X, X)} | \pi^{2n} L \pi^{\infty(3n)} |_{\mathscr{B}(X, X)} \\
&=: Z_1.
\end{align*}
\item
By the triangle inequality, we obtain
\begin{align*}
&\sup_{x \in \textnormal{cl} (B_R(x_0))} | \mathcal{A} D^2 F(x) |_{\mathscr{B}(X^2, X)} \\
&= \sup_{x \in \textnormal{cl} (B_R(x_0))} | ( A\pi^{2n} + L|_{\pi^{\infty(2n)}X}^{-1} \pi^{\infty(2n)} ) D^2 F(x) |_{\mathscr{B}(X^2, X)} \\
&\leq  (| A |_{\mathscr{B}(X, X)} + | L|_{\pi^{\infty(2n)}X}^{-1} \pi^{\infty(2n)} |_{\mathscr{B}(X, X)}) \sup_{x \in \textnormal{cl} (B_R(x_0))} | D^2 F(x) |_{\mathscr{B}(X^2, X)} \\
&\leq  (| A |_{\mathscr{B}(X, X)} + | L|_{\pi^{\infty(2n)}X}^{-1} \pi^{\infty(2n)} |_{\mathscr{B}(X, X)}) \zeta \\
&=: Z_2,
\end{align*}
for some $\zeta = \zeta(x_0, R) \ge \sup_{x \in \textnormal{cl} (B_R(x_0))} | D^2 F(x) |_{\mathscr{B}(X^2, X)}$ obtained in practice by applying the triangle inequality.

\end{enumerate}
Consider the fixed-point operator $T : X \to X$ given by $T(x) := x - \mathcal{A} F(x)$, which is twice Fr\'{e}chet differentiable.
By hypothesis, there exists $r_0 \in [0, R]$ satisfying the bounds \eqref{prop4.3:bounds}. 
On the one hand, a second-order Taylor expansion of $T$ yields
\begin{equation*}
    |T(x) - x_0|_X \leq Y + Z_1 r_0 + \frac{Z_2}{2} r_0^2 \leq r_0, \quad \text{for all } x \in \textnormal{cl}(B_{r_0}(x_0)).
\end{equation*}
On the other hand, the Mean Value Theorem implies
\begin{equation*}
    | T(x_1) - T(x_2)|_X \leq (Z_1 + Z_2 r_0) | x_1 - x_2 |_X < | x_1 - x_2 |_X, \quad \text{for all } x_1, x_2 \in \textnormal{cl}(B_{r_0}(x_0)).
\end{equation*}
Thus, $T$ satisfies the Banach Fixed-Point Theorem in $\textnormal{cl}(B_{r_0}(x_0))$: there exists a unique $\tilde{x} \in \textnormal{cl}(B_{r_0}(x_0))$ such that $T(\tilde{x}) = \tilde{x}$. By injectivity of $\mathcal{A}$, it follows that $F(\tilde{x}) = 0$.

Finally, we want to prove that $\tilde{x} := (\tilde{a}_1, \tilde{a}_2, \tilde{a}_3, \tilde{a}_4, \tilde{a}_5, \tilde{\gamma}, \tilde{\eta}_1, \tilde{\eta}_2)$ satisfy $\tilde{\gamma}, \tilde{\eta}_1, \tilde{\eta}_2 \in \mathbb{R}$ and $\tilde{a}_1, \tilde{a}_2, \tilde{a}_3, \tilde{a}_4, \tilde{a}_5$ are sequences of Fourier coefficients of real Fourier series. Mathematically, this means that $\tilde{\gamma} = \overline{\tilde{\gamma}}$, $\tilde{\eta}_1 = \overline{\tilde{\eta}_1}$, $\tilde{\eta}_2 = \overline{\tilde{\eta}_2}$ and $(\tilde{a}_j)_k = \overline{(\tilde{a}_j)_{-k}}$ for all $k \in \mathbb{Z}$, $j = 1, \dots, 5$. Note that the overline represents the complex conjugacy.

For convenience, we introduce the symbol $\dagger : X \to X$ defined by
\begin{equation*}
\begin{aligned}
\dagger(x) = \left(\left\{ \overline{(a_1)_{-k}} \right\}_{k \in \mathbb{Z}}, \left\{ \overline{(a_2)_{-k}} \right\}_{k \in \mathbb{Z}}, \left\{ \overline{(a_3)_{-k}} \right\}_{k \in \mathbb{Z}}, \left\{ \overline{(a_4)_{-k}} \right\}_{k \in \mathbb{Z}}, \left\{ \overline{(a_5)_{-k}} \right\}_{k \in \mathbb{Z}}, \overline{\gamma}, \overline{\eta_1}, \overline{\eta_2}\right), \\
\text{for all } x = (a_1, a_2, a_3, a_4, a_5, \gamma, \eta_1, \eta_2) \in X.
\end{aligned}
\end{equation*}
In particular, if $x = \dagger(x)$ for some $x = (a_1, a_2, a_3, a_4, a_5, \gamma, \eta_1, \eta_2) \in X$, then $a_1, a_2, a_3, a_4, a_5$ are Fourier coefficients of real functions.

By assumption $x_0 = \dagger(x_0)$. Note that $| \dagger(\tilde{x}) - x_0 |_X = | \dagger(\tilde{x}) - \dagger(x_0) |_X = | \dagger(\tilde{x} - x_0) |_X = | \tilde{x} - x_0 |_X$, i.e. $\dagger(\tilde{x}) \in \textnormal{cl}(B_{r_0}(x_0))$. Also, it holds that $F(\dagger(\tilde{x})) = \dagger(F(\tilde{x})) = 0$. Whence, by local uniqueness, $\tilde{x} = \dagger(\tilde{x})$ which concludes the proof.
\end{proof}


Next, to complete (ii) in Section \ref{sec:typeVIIandIX_Proofs}, consider a numerical approximation given by $x_0 := (a_{0,1}, a_{0,2}, a_{0,3}, a_{0,4}, a_{0,5}, \gamma_0, 0, 0) \in \pi^n X$, where the truncation operator $\pi^n$ is defined in \eqref{eq:proj}-\eqref{eq:proj2}, such that $\gamma_0 \in \mathbb{R}$ and $a_{0,1}, a_{0,2}, a_{0,3}, a_{0,4}, a_{0,5}$ are sequences of Fourier coefficients of real Fourier series. Assume that Proposition \ref{prop:radii_polynomial} is satisfied: there exists $r_0 \in [0, +\infty)$ and thus a unique $\tilde{x} := (\tilde{a}_1, \tilde{a}_2, \tilde{a}_3, \tilde{a}_4, \tilde{a}_5, \tilde{\gamma}, \tilde{\eta}_1, \tilde{\eta}_2) \in \textnormal{cl}(B_{r_0}(x_0))$ such that $F(\tilde{x})=0$ and $\tilde{\gamma}, \tilde{\eta}_1, \tilde{\eta}_2 \in \mathbb{R}$, $\tilde{a}_1, \tilde{a}_2, \tilde{a}_3, \tilde{a}_4, \tilde{a}_5$ are sequences of Fourier coefficients of real Fourier series.

We now detail the algorithm to verify that $\Sigma_+(t) := \sum_{k \in \mathbb{Z}} (\tilde{a}_2)_k e^{i k t}$ is not identically zero and $N_1 N_3 (t) := \sum_{k \in \mathbb{Z}} (\tilde{a}_3 * \tilde{a}_5)_k e^{i k t}$ is not identically zero and does not change sign.

To evaluate $\Sigma_+$, $N_1$ and $N_3$ rigorously, we use the interval enclosures
\begin{subequations}
\begin{align}
\label{eq:Sigma_enclosure} \Sigma_+(t) &\in \left(\sum_{|k| \leq n} (a_{0,1})_k e^{i k t}\right) + [- r_0, r_0], \qquad t \in \mathbb{R},\\
\label{eq:N_enclosure} N_j(t) &\in \left(\sum_{|k| \leq n} (a_{0,j+2})_k e^{i k t}\right) + [- r_0, r_0], \qquad j\in\{1,3\}, \qquad  t \in \mathbb{R}.
\end{align}
\end{subequations}
Firstly, to establish that $\Sigma_+$ and $N_1N_3$ are not identically zero, we choose some times $t_1,t_2,\in[0,2\pi]$ and compute the interval enclosures of $\Sigma_+(t_1)$ (via \eqref{eq:Sigma_enclosure}), $N_1(t_2)$ and $N_3(t_2)$ (via \eqref{eq:N_enclosure}). If these intervals do not contain zero, then it must be the case that $\Sigma_+(t_1)\neq 0$ and $N_1(t_2)N_3(t_2)\neq 0$.

Secondly, we show that $N_1 N_3$ has constant sign in $[0, 2\pi)$. It is sufficient to prove that each of $N_1$ and $N_3$ does not have a zero in $[0,2\pi]$. Note that using \eqref{eq:auxiliary_odeN1}, \eqref{eq:auxiliary_odeN3} and Proposition \ref{prop:radii_polynomial}, we can rigorously bound the derivatives $N_1'$ and $N_3'$ as follows:
\begin{align}
|N_1'|_\infty = \sup_{t\in[0,2\pi]}|N_1'(t)| &\leq 4(|\gamma_0| + r_0)(|a_{0,1}|_{X_\textnormal{Fourier}}+r_0)(|a_{0,3}|_{X_\textnormal{Fourier}}+r_0),\nonumber\\
|N_3'|_\infty = \sup_{t\in[0,2\pi]}|N_3'(t)| &\leq 2(|\gamma_0| + r_0)(|a_{0,5}|_{X_\textnormal{Fourier}}+r_0)(|a_{0,1}|_{X_\textnormal{Fourier}}+r_0 + \sqrt{3}(|a_{0,2}|_{X_\textnormal{Fourier}}+r_0)).\nonumber
\end{align}
To verify that $N_1$ and $N_3$ have no zeros, we use an inductive argument; for clarity, we present the argument only for $N_1$ (the case of $N_3$ is treated similarly). Set $t_0 := 0$ and evaluate the enclosure \eqref{eq:N_enclosure} for $N_1(t_0)$. If this enclosure does not contain zero, then $N_1(t_0)\neq 0$. Since $|N_1'|_\infty$ is a Lipschitz constant for $N_1$, there can be no zero in the interval $[t_0,t_0 + |N_1(t_0)|/|N_1'|_\infty)$. Therefore, define the time iterates
\begin{equation*}
t_j := t_{j-1} + \frac{\left|\sum_{|k|\leq n} (a_{0,3})_k e^{i k t_{j-1}}\right| - r_0}{|N_1'|_\infty}, \qquad j \in \mathbb{N}.    
\end{equation*}
By construction, $N_1$ does not have a zero in $[t_0,t_1]$. As an inductive hypothesis, suppose $N_1$ does not have a zero in $[t_0,t_j]$ for some $j \in \mathbb{N}$. We verify that $N_1(t_j)\neq 0$ using an inclusion check with \eqref{eq:N_enclosure}. Then, by construction, the interval $[t_0,t_{j+1}]$ will not contain zero. We halt the iteration as soon as some $j^* \in \mathbb{N}$ is reached such that $t_{j^*} \geq 2\pi$.

\textbf{Acknowledgments.} 
PL was funded by CNPq, 163527/2020-2 and 160956/2022-6. JPL was funded by NSERC.

\textbf{Data Availability Statement.} 
We confirm that all relevant data are included in the article.

\bibliographystyle{plain}

\end{document}